\documentclass[11pt]{article}

\usepackage{amsmath}
\usepackage{enumitem}
\setlist[description]{style=nextline, labelindent=0pt}
\usepackage{amssymb}
\usepackage{amsfonts}
\usepackage{textcomp}
\usepackage{algorithm}
\usepackage{authblk}
\usepackage[noend]{algpseudocode}
\usepackage{amsthm}
\usepackage{xspace}
\usepackage{fullpage}
\usepackage[english]{babel}
\usepackage[utf8]{inputenc}
\let\originalS\S 
\usepackage[classfont=bold]{complexity}
\renewcommand{\S}{\originalS}
\usepackage[pagebackref=false, colorlinks=true, citecolor=DarkGreen, linkcolor=Navy]{hyperref}
\usepackage{mathtools}
\usepackage[svgnames]{xcolor}
\usepackage{bussproofs}
\usepackage{subfig}
\usepackage{todonotes}
\usepackage{thmtools}
\usepackage{thm-restate}
\usepackage[capitalise, noabbrev]{cleveref}
\usepackage{bm}
\usepackage{bbm}
\usepackage[theorems,skins]{tcolorbox}

\usepackage[T1]{fontenc}
\usepackage{lmodern}

\usepackage{titlesec}
\titleformat{\section}
  {\large\bfseries} %
  {\thesection}{1em}{}
\titleformat{\subsection}
  {\bfseries} %
  {\thesubsection}{1em}{}

\usepackage[%
activate={true,nocompatibility},
selected=true,
tracking=true,
factor=1100,
babel=true
]%
{microtype}

\usepackage
[%
labelformat=default, %
labelsep=colon, %
textformat=simple, %
font={small,sf}, %
labelfont=bf %
]%
{caption}

\patchcmd{\thebibliography}
  {\labelsep}
  {\labelsep \itemsep=0pt \parsep=0pt \parskip=0pt}
  {}{}

\makeatletter 
 \renewenvironment{abstract}{%
   \small%
   \begin{center}%
     {\bfseries %
     \abstractname\vspace{-.5em}\vspace{\z@}}%
   \end{center}%
   \quotation}
 {\endquotation}
\makeatother

\theoremstyle{plain}
\newtheorem{theorem}{Theorem}[section]
\newtheorem{lemma}[theorem]{Lemma}
\newtheorem{corollary}[theorem]{Corollary}
\newtheorem{proposition}[theorem]{Proposition}

\theoremstyle{definition}
\newtheorem{remark}[theorem]{Remark}

\newtheorem{definition}[theorem]{Definition}

\newtheorem{fact}[theorem]{Fact}

\newtheorem*{integrality-assumption}{Lemma~\ref{lem:ipr-alekhnovich-razborov}}

\newtheorem*{reducibility-gives-reduction-operator}{Lemma~\ref{lem:reducibility-gives-reduction-operator}}

\usepackage{todonotes}

\newcommand{\newauthor}[3]{
    \newcounter{#1comment}
    \setcounter{#1comment}{1}
    \expandafter\newcommand\csname #1\endcsname[1]{%
            \todo[size = \footnotesize, backgroundcolor = {#3}, caption = {}]{
                \arabic{#1comment}:
                {##1} --~\textbf{#2}
            }
            \addtocounter{#1comment}{1}
    }
    \expandafter\newcommand\csname #1changed\endcsname[1]{%
        \ifdraft
            \colorbox{#3}{
                ##1
            }%
        \else
            ##1%
        \fi
    }
}

\usepackage{fancyhdr}

\newauthor{jc}{Jonas}{orange}
\newauthor{shp}{Shuo}{yellow}
\newauthor{yg}{Yassine}{LimeGreen}

\newcommand{\PC}{\mathrm{PC}}
\newcommand{\DPC}{\operatorname{Deg}_{\PC,\F}}
\newcommand{\res}[2]{#1\!\upharpoonright_{#2}}
\newcommand{\Vars}{\mathrm{Vars}}

\DeclareMathOperator{\LM}{LM}
\DeclareMathOperator{\ML}{ML}

\DeclareMathOperator{\OR}{OR}
\DeclareMathOperator{\AND}{AND}
\DeclareMathOperator{\NOT}{NOT}
\renewcommand{\MOD}{\mathrm{MOD}}
\DeclareMathOperator{\App}{App}
\DeclareMathOperator{\Rule}{R}

\newcommand{\intersection}{\cap}

\newcommand{\union}{\cup}

\newcommand{\eqperiod}{\enspace .}
\newcommand{\eqcomma}{\enspace ,}

\newcommand{\expectation}{\mathbb{E}}
\newcommand{\N}{\mathbb{N}}

\newcommand{\F}{\mathbb{F}}

\renewcommand{\epsilon}{\varepsilon}

\renewcommand{\epsilon}{\varepsilon}
\newcommand{\nbh}[1]{N(#1)}

\newcommand{\pseudred}{R}

\newcommand{\nvars}{n}

\newcommand{\ideal}[1]{\langle #1 \rangle}
\newcommand{\redop}{R}

\newcommand{\degstd}{D}

\newcommand{\setP}{P}

\newcommand{\vars}[1]{\operatorname{Vars}(#1)}

\newcommand{\cl}[1]{\operatorname{cl}(#1)}

\newcommand{\polyp}{p}
\newcommand{\polyq}{q}

\newcommand{\refstd}{\pi}

\newcommand{\B}{\mathbb B}
\newcommand{\im}{\operatorname{im}}
\newcommand{\Res}[1]{\operatorname{Res}(#1)}

\newcommand{\polyspace}[1]{V_{#1}}
\newcommand{\rankshort}{\mathrm{rk}}
\newcommand{\rules}{\mathcal{R}}
\newcommand{\Resol}{\operatorname{Resol}}
\newcommand{\sdpolyp}{f}

\providecommand{\sdFp}{\mathbb{F}_p}
\providecommand{\sdPC}{\operatorname{Deg}_{\mathrm{PC},\mathbb{F}_p}}
\providecommand{\sdAC}{\mathrm{AC}^{0}[p]}
\providecommand{\sdMOD}{\mathrm{MOD}}
\providecommand{\sdBool}{\mathrm{Bool}}
\providecommand{\sdML}{\mathrm{ML}}

\providecommand{\sdRank}{\operatorname{rank}}
\providecommand{\sdCD}{\PC_{\dthresh}}

\providecommand{\sdE}{\mathbb{E}}

\providecommand{\dthresh}{D}
\newcommand{\ACzero}{\mathsf{AC}^0}

\DeclareMathOperator{\mult}{\operatorname{mult}}

\usepackage[letterpaper,
   top=1in,
    bottom=1in,
    left=1in,
    right=1in
]{geometry}

\title{Random ($\log{n}$)-CNF Formulas Are Hard for Depth-$2$ Frege with Modular Gates}

\author{Jonas Conneryd\thanks{Lund University. \texttt{\href{mailto:jonas.conneryd@cs.lth.se}{jonas.conneryd@cs.lth.se}
}}}

\date{\today}

\begin{document}
\maketitle

\thispagestyle{empty}

\begin{abstract}
    \noindent
     We prove that, with high probability, refuting a random $k$-CNF formula on $n$ variables with ${M = n2^{k}}$ clauses for $\omega_n(1) \leq k \leq O(\log{n})$ requires polynomial calculus degree at least ${n/2-O(n/\sqrt{k})}$ over every field. Through Pang's recent tradeoff between $\ACzero[p]$\nobreakdash-Frege size and polynomial calculus degree~\cite[Section~8]{Pang26SizeDegree}, this implies a size lower bound of $\exp(\Omega_p(\sqrt{k} \log{n}))$ for depth-$2$ $\ACzero[p]$-Frege. In this range of $k$, the bound is superpolynomial in $M$ while $M \leq \poly(n)$. 
     We also prove a nearly matching degree upper bound of $n/2-n\ln(k)/{2k} +O(n/k)$, which shows that other formula families are needed to obtain lower bounds for full $\ACzero[p]$-Frege through Pang's tradeoff.

\end{abstract}

\flushbottom

\paragraph{AI Disclosure:} The proofs of the results in this paper were mainly conceived by ChatGPT-6 Astra and ChatGPT-6.1 Sol after suggestions by and resulting discussions with the author. Both models also contributed to the writing and editing of the manuscript. The author checked and revised all AI-generated results, and takes full responsibility for the content of the paper.

\newpage
\thispagestyle{empty}
\tableofcontents

\newpage

\setcounter{page}{1}

\section{Introduction}
\label{sec:introduction}

Ever since the seminal work of Cook and Reckhow~\cite{CR79Relative} that founded the field of proof complexity, the so-called \emph{Frege} proof systems have played a fundamental role in this area. These systems reason using any sound and complete set of propositional inference rules, and each line in a \emph{$\mathcal{C}$-Frege} proof is a Boolean circuit from the class $\mathcal{C}$. 

Historically, proof complexity lower bounds for $\mathcal{C}$-Frege have followed developments in lower bounds for $\mathcal{C}$-circuits. Perhaps the most notable example is when $\mathcal{C}$ is the class $\ACzero$, which consists of polynomial-size constant-depth circuits with gates in the de Morgan basis. In the 1980s, a celebrated series of works~\cite{Ajtai83Formulae, FSS84ParityCircuits, Yao85Separating,  Hastad86AlmostOptimal} culminated in nearly optimal lower bounds for low-depth circuits computing, e.g., the parity function. These works identified the power of \emph{random restrictions} for simplifying $\ACzero$ circuits, as captured by so-called \emph{switching lemmas}. Building on these works as well as more recent strengthenings such as~\cite{Hastad14Correlation, HRST17AverageCase}, a long sequence of papers, e.g.,~\cite{Ajtai94Complexity, PBI93ExponentialLowerBounds, KPW95ExponentialLowerBound, PRST16Polylogarithmic, Hastad20SmallDepth, RPT22Tradeoffs, HR25BoundedDepth, Hastad25PHP}, has successfully ported these powerful techniques to $\ACzero$-Frege and obtained strong lower bounds for this proof system. 

Since parity does not simplify under random restrictions, a natural next step in circuit complexity was to allow the circuit gates to compute parity functions or, more generally, modular predicates. The resulting class of polynomial-size constant-depth circuits with AND, OR, NOT, or $\MOD_{p, b} \coloneqq \mathbbm{1}[\text{sum of inputs is $b\bmod{p}$}]$ gates is denoted $\ACzero[p]$. For prime~$p$, lower bounds for $\ACzero[p]$ circuits were shown by Razborov~\cite{Razborov87LowerBounds} and Smolensky~\cite{Smolensky87Algebraic} using the so-called \emph{approximation method}. Here, the idea is to approximate each circuit gate by a low-degree polynomial over the finite field $\F_p$ and then use the depth assumption to build a good low-degree polynomial approximation of any function computed by an $\ACzero[p]$ circuit. By showing that a Boolean function $f$ cannot be approximated by low-degree polynomials, one thus concludes that $f\not\in \ACzero[p]$.  

Although the switching lemmas involved in lower bounds for $\ACzero$-Frege are notoriously technical, the method of random restrictions is fundamentally compatible with proof complexity since reasonable proof systems retain valid inferences when restricting variables in each proof line. By contrast, the approximation method seems much harder to implement in this setting, and therefore it has been an open problem since \cite{Razborov87LowerBounds, Smolensky87Algebraic} to prove lower bounds for $\ACzero[p]$-Frege. %
There have, however, been attempts to do so. 
Buss et al.~\cite{BIKPRS97ProofCplx} introduced a version of the approximation method into proof complexity already in the 1990s: they showed that lower bounds for $\ACzero[p]$-Frege follow from degree lower bounds for the algebraic proof system \emph{Nullstellensatz}~\cite{BIKPP94LowerBounds} augmented with extension axioms that, roughly speaking, formalize the approximation method. This direction has been further explored by, e.g., Kraj{\'i}{\v{c}}ek~\cite{Krajicek15Reduction, Krajicek24ExtendedNullstellensatz, Krajicek26PseudoSolutions} but has thus far not led to $\ACzero[p]$-Frege lower bounds.

Very recently, there has been exciting progress in this line of work. In an AI-assisted breakthrough, Braun~\cite{Braun26LowerBound} proved the first superpolynomial (in fact, exponential) lower bounds for the \emph{resolution over parities} proof system~\cite{IS14Resolution}, also denoted $\Res{\oplus}$, where each line is a disjunction of $\F_2$-linear equations (and hence a particularly simple $\ACzero[2]$-circuit). His lower bound is for the so-called \emph{bit pigeonhole principle}. His proof uses extension axioms as in~\cite{BIKPRS97ProofCplx} to replace each line in a small $\Res{\oplus}$ proof by a low-degree polynomial. He then argues that there is a variable substitution that removes all extension variables of lines that were previously low-rank. To deal with high-rank lines, he introduces the \emph{common-multiplier method}: he shows that there is a polynomial $f$ of not-too-high degree such that all $f$-multiples of high-rank lines are derivable in low degree from the input axioms in the \emph{polynomial calculus} proof system~\cite{CEI96Groebner,ABRW02SpaceComplexity}, while~$f$ itself is not. Thus, multiplying all lines of the substituted refutation by $f$, he obtains a low-degree polynomial calculus derivation of $f$---a contradiction. To argue that such an $f$ exists, Braun extends Razborov's particularly strong polynomial calculus degree lower bound for the pigeonhole principle~\cite{Razborov98LowerBound}.

Shortly thereafter, the common-multiplier method was adapted and strengthened by Pang \cite[Section~8]{Pang26SizeDegree} in another AI-assisted paper. The approach in \cite{Pang26SizeDegree} instead replaces each line $L$ in an $\ACzero[p]$-Frege proof directly by (the negation of) its Razborov--Smolensky approximation $\App_L$ in the original variables. Each inference $L_1, \ldots, L_s \vdash L$ then becomes an approximate inference $\App_{L_1}, \ldots, \App_{L_s} \widetilde{\vdash} \App_{L}$ with an associated \emph{inference error polynomial $E$}, which is $1$ when the approximated premises are true but the approximated conclusion is false. Pang then argues that in a small proof, a strong enough polynomial calculus degree lower bound yields a low-degree common multiplier $f$ that turns all $f$-multiples of the error polynomials into low-degree consequences of the input axioms, while $f$ itself is not derivable in low degree. As in~\cite{Braun26LowerBound}, one  can then convert the proof into a low-degree polynomial calculus derivation of $f$ to obtain a contradiction. This yields a \emph{generic} size-degree tradeoff, showing that strong enough polynomial calculus degree lower bounds imply $\ACzero[p]$-Frege lower bounds in a black-box fashion.

Currently, the drawback is that the parameters in Pang's tradeoff are very restrictive: for $\ell \geq 3$ and a formula $F$ with $n$ variables, his results imply superpolynomial lower bounds for refuting $F$ in depth-$\ell$ $\ACzero[p]$-Frege whenever the polynomial calculus degree required to refute $F$ is at least $n/2 - o(n/\log^{\ell-2}(n))$; see \cref{thm:sd-fixed-depth}. For $\ell \geq 3$, such strong degree lower bounds are not known for any $\poly(n)$-size formula.\footnote{The same arguments as for \cref{thm:main} applied to suitably large $k \gg \log{n}$ let us go even beyond degree $n/2$ for superpolynomially large randomly sampled formulas. However, any minimally unsatisfiable subformula then turns out to be larger than the lower bound obtained from Pang's tradeoff, making the result trivial.} However, it turns out that for depth-$2$ $\ACzero[p]$-Frege, a degree lower bound of $n/2 - o(n)$ suffices. Such a bound is the main result of this paper. 

\begin{theorem}\label{thm:main}
For every sufficiently large integer $n$ and $\omega_n(1)\leq k \leq O(\log{n})$, it holds with probability $1-o_n(1)$ that a random $k$-CNF formula on $n$ variables with $M=n2^k$ clauses is unsatisfiable but requires polynomial calculus degree strictly greater than $n/2-4n/\sqrt{k}$ to refute over every field. 
\end{theorem}
As an immediate corollary of \cref{thm:main} and Pang's tradeoff, we obtain superpolynomial size lower bounds for depth-$2$ $\ACzero[p]$-Frege. 
\begin{corollary}\label{cor:size-lower-bound}
    For any prime $p$, with probability $1-o_n(1)$ a random $n$-variable $k$-CNF formula as in \cref{thm:main} requires size $\exp(\Omega_p(\sqrt{k} \log{n}))$ to refute in depth-$2$ $\ACzero[p]$-Frege. 
\end{corollary}
Pang's tradeoff is established for a specific set of inference rules. By Reckhow's Theorem~\cite{Reckhow75Thesis} (see also \cite[Chapter~2.3]{Krajicek19ProofComplexity}), Frege systems over different sets of inference rules polynomially simulate each other, but at the cost of an additive increase in depth. Since we want to establish depth-$2$ $\ACzero[p]$-Frege lower bounds for any fixed rule set, we provide a straightforward extension of Pang's tradeoff to this setting in Appendix~\ref{sec:sd}. This can be inferred from Pang's proof, but we give the full proof for completeness. The statement we need for \cref{cor:size-lower-bound} is recorded there as \cref{thm:sd-depth-two}. 

As previously noted, for $\ell \geq 3$ the degree lower bound needed for Pang's result to apply to depth-$\ell$ $\ACzero[p]$-Frege is $n/2-o(n/\log^{\ell-2}(n))$. Unfortunately, such lower bounds cannot be obtained from applying Pang's tradeoff to random $k$-CNF formulas as in \cref{thm:main} due to the following upper bound.
\begin{theorem}\label{thm:upper-bound}
With probability $1-o_n(1)$, an $n$-variable random $k$-CNF formula as in \cref{thm:main} can be refuted in polynomial calculus degree $n/2 -n\ln(k)/2k + O(n/k)$ over every field.
\end{theorem}

\subsection{Proof Overview}\label{sec:proof-overview}

\paragraph{Lower bounds.}  

We prove the lower bound in \cref{thm:main} by constructing a so-called \emph{degree-$\degstd$ pseudo-reduction operator}~\cite{Razborov98LowerBound} for a random $k$-CNF formula $F$. This is an $\F$-linear operator $\pseudred$ on polynomials that sends the constant monomial $1$ to $1$ but sends all polynomials derivable from the input in degree $\degstd$ to $0$; hence, if such an operator exists there can be no degree-$\degstd$ refutation.
To construct $\pseudred$, we use the \emph{Alekhnovich--Razborov method}~\cite{AR03LowerBounds} based on the notion of \emph{reduction operators} from Gr\"obner basis theory. 

For a set $\setP$ of polynomials and a so-called \emph{admissible} total ordering $\prec$ on monomials (where in particular $1$ is smallest), the reduction operator $R^\prec_{\ideal{P}}$ applied to a polynomial $p$ finds a polynomial $q$ in the ideal $\ideal{\setP}$ generated by $\setP$ whose largest monomial with respect to $\prec$ is in~$p$, subtracts $q$ from $p$, and then iterates on $p-q$ until there is no such $q$, similarly to a matrix being reduced to row-echelon form. The idea behind the Alekhnovich--Razborov method is that if the formula $F$ were satisfiable, then $R_{\ideal{F}}$ would be the best possible pseudo-reduction operator: $1$ is then not in $\ideal{F}$ and therefore mapped to $1$, while all polynomials derived from $F$, of any degree, are in $\ideal{F}$ and hence sent to $0$. A pseudo-reduction operator for unsatisfiable $F$ should therefore resemble $\pseudred_{\ideal{F}}^\prec$ as closely as possible while sending $1$ to $1$.

Toward this goal, Alekhnovich and Razborov associate to each monomial $m$ a \emph{satisfiable} subset $S(m) \subseteq F$ and define $\pseudred$ on monomials by $\pseudred(m) \coloneqq \smash{R^\prec_{\ideal{S(m)}}(m)}$, linearly extended to arbitrary polynomials. Since each $S(m)$ is satisfiable, we automatically get $R(1)=1$ with this definition. The challenge is to pick $S(m)$ so that (i) each input polynomial $p_C$ for $C\in F$ is sent to $0$, and (ii) $R$ is compatible with multiplication, in the sense that $R(x_im) = R(x_iR(m))$ whenever $\deg{m} \leq \degstd -1$. 

As in many previous works, the property of random $k$-CNF formulas on $n$ variables that we use to show (i) and (ii) is that their \emph{constraint-variable incidence graphs (CVIGs)}---that is, the bipartite graph with clauses on the left, variables on the right, and edges whenever a variable appears as a literal in a clause---are excellent \emph{expanders}. This property takes many concrete forms depending on the application, but in our case it can be formulated as follows: ``small'' subformulas of $F$ of size smaller than $s=n/\sqrt{k}$ are satisfiable, while ``medium-size'' subformulas in the interval $s$ to $2s$ have large neighborhoods in the CVIG.

With this notion of expansion, we formulate our lower bound arguments using the well-known notion of \emph{(matching) deficiency}~\cite[Chapter~14.2]{HandbookSAT21}. For a variable set $T$, the deficiency of a clause set $A$ with respect to $T$ is the quantity
\begin{equation}
  \label{eq:definciency}
  h_T(A)=|A|-|\Vars(A)\setminus T|\eqperiod
\end{equation}
A set $A$ with positive deficiency with respect to $T$ is an obstruction to Hall's condition being satisfied in the CVIG $G_F$ of $F$ with the variables in $T$ removed. Note that if Hall's condition were satisfied in $G_F$, then $F$ would be satisfiable since assigning each literal to satisfy their matched clause is a consistent satisfying assignment for $F$. 

Using expansion, we can show that all medium-sized sets have negative deficiency with respect to any variable set $T$ of size at most $n/2-4s$. Thus, any obstruction to Hall's condition holding locally in subformulas of size at most $2s$ after removing $T$ must be small. Using this observation, we can find a small clause set called \emph{closure} of $T$ and denoted $\cl{T}$ which, roughly speaking, simultaneously captures all obstructions to Hall's condition with respect to $T$ on medium-size subformulas. In more detail, we show that for any superset $B\supseteq \cl{T}$ of size at most $2s$, every nonempty subset of $B\setminus \cl{T}$ has negative deficiency with respect to $T\union \vars{\cl{T}}$. Hence, by Hall's condition, we can satisfy $B\setminus \cl{T}$ without assigning any variable in this variable set. Thus, if we let the set $S(m)$ in the definition of $\pseudred$ equal $\cl{m} \coloneqq \cl{\vars{m}}$, we have for any such $B$ that \smash{$\pseudred^\prec_{\ideal{\cl{m}}}(m) = R^\prec_{\ideal{B}}(m)$} since a partial assignment $\rho$ satisfying $B\setminus \cl{T}$ sets the generators of $B\setminus \cl{T}$ to $0$ while leaving $\vars{m} \union \vars{\cl{m}}$ untouched. 

With this fact established, we can prove (i) and (ii) for polynomials of degree at most $n/2 - 4s$ by a set of arguments originating in~\cite{AR03LowerBounds} and further refined in, e.g., \cite{GL10Optimality,MN24Generalized}. In both cases, the idea is to consider $\pseudred$ applied to a polynomial $p = \sum_i a_i m_i$ and replace each ideal in the summands of 
\begin{equation}
  \label{eq:r-op-evaluation}
  R(p) = \sum_{i} a_i R^\prec_{\ideal{\cl{m_i}}}(m_i)
\end{equation}
by a single (satisfiable) ideal $\ideal{B} \supseteq \bigcup_{i} \ideal{\cl{m_i}}$, so that \smash{$R(p) = R^\prec_{\ideal{B}}(p)$} is an actual reduction modulo $\ideal{B}$. For~(i), $p$ is a clause polynomial $p_C$ for $C\in F$, and then we can take $B = \cl{\vars{p_C}}$. For~(ii), $p$ is instead the polynomial $x_iR(m)$ for a monomial $m$ with $\deg{m} \leq \degstd-1$. In this case we can show that $B = \cl{x_im}$ suffices.

\paragraph{Upper bounds.} The upper bound in \cref{thm:upper-bound} is based on a deterministic criterion. Say that a clause set~$G$ of width at most $\degstd$ \emph{covers the $\degstd$-element variable sets} if for each multilinear monomial $m$ of degree at least $\degstd$, there exists a clause $C\in G$ such that the monomial $m_C \coloneqq \prod_{x \in \vars{C}} x$ divides $m$. We show that if $G$ is unsatisfiable and covers the $\degstd$-element variable sets, then we can derive $1$ from $G$ in degree at most~$\degstd+1$. The proof idea is, roughly speaking, to perform Gaussian elimination in the vector space $\PC_{\degstd +1}$ of polynomials derivable from $G$ in degree at most $\degstd +1$, in a manner quite similar to the algorithm in~\cite[Section~4]{CEI96Groebner}. By the covering hypothesis, every polynomial of degree at most $\degstd+1$ reduces to one of degree at most $\degstd-1$. As we show in \cref{ub:covering}, this bound, the unsatisfiability of $G$, and the fact that $\PC_{ \degstd +1 }$ is closed under low-degree variable multiplications together force $1$ to eventually appear in the elimination.

Given this result, we prove that with high probability, the set of (polynomial representations of) resolvents of clauses in a random $k$-CNF formula $F$ covers the $\degstd$-element sets for $\degstd = n/2 -n\ln(k)/2k + O(n/k)$. This follows from a rather simple argument where we sample the clauses of $F$ in two equally large batches. After sampling both batches, we match clauses between batches by resolution compatibility and show that the resulting resolvents cover all $\degstd$-element variable sets with high probability. Since the set of resolvents is derivable in low degree from the polynomial representation of $F$, the upper bound follows.

\subsection{Related Work}
\label{sec:related-work}
The proof of \cref{thm:main} follows techniques developed in a long line of work investigating degree lower bounds for algebraic proof systems. The first superconstant lower bounds for Nullstellensatz were proved in \cite{BIKPP94LowerBounds}, which also introduced the proof system. The first superconstant degree lower bounds for polynomial calculus were proved by Razborov~\cite{Razborov98LowerBound} (who also achieved optimal, linear lower bounds in the same paper), followed by \cite{BGIP01LinearGaps, BI10Random} where techniques based on binomial ideals were developed. Based on a similar philosophy as~\cite{Razborov98LowerBound}, Alekhnovich and Razborov~\cite{AR03LowerBounds} in turn introduced a more general approach, which is used in this paper. These techniques were further developed and generalized by, e.g., \cite{GL10Automatizability, GL10Optimality,Filmus19AlekhnovichRazborov, MN24Generalized,CdRNPR25GraphColouring}. (Variants of) expansion of the constraint-variable incidence graph is the same property as in previous polynomial calculus degree lower bounds for random CNF formulas~\cite{BI10Random,AR03LowerBounds,MN24Generalized} as well as in lower bounds for random formulas in many other proof systems, e.g.,~\cite{CS88ManyHard,SBI04SwitchingLemma,Alekhnovich11LowerBoundsk-DNF3-CNF, Sokolov24Random}. Arguments relating expansion to matchings in the constraint-variable incidence graph after removing a closure go back at least to \cite{AR03LowerBounds,ABRW04Pseudorandom} and are by now well known.

There is also a line of work investigating polynomial calculus over variable domains different from $\{0,1\}$. Since $\ACzero[p]$-Frege can perform variable substitutions from $\{0, 1\}$ to other domains in $\F_p$ within the proof system, proving lower bounds over other domains is a necessary step toward $\ACzero[p]$\nobreakdash-Frege lower bounds. The first size lower bounds in this setting were shown by Sokolov~\cite{Sok20SemiAlgebraic}, and follow-up work by Impagliazzo, Mouli, and Pitassi~\cite{IMP23PCLowerBounds} and Dahiya, Mahajan, and Mouli~\cite{DMM24NewLowerBounds} establishes lower bounds for all variable domains over $\F_p$.\footnote{The lower bounds there are stated for the more general setting of partially overlapping extension variables, but some of these stronger results have been retracted~\cite{Mouli26AcknowledgementIMP23, Mouli26AcknowledgementDMM24}.} These results seem incomparable to those in this paper since it is not clear how to perform variable substitutions in depth $2$. 

Braun's lower bound for resolution over parities~\cite{Braun26LowerBound} is in turn the culmination of a long line of works investigating fragments and variants of this proof system. Raz and Tzameret~\cite{RT08ResolutionLinear} studied resolution over integer linear equations and proved lower bounds for a fragment of this proof system. Itsykson and Sokolov~\cite{IS14Resolution} defined the version over $\F_2$ and proved lower bounds for tree-like such proofs, which were extended to a stronger fragment by Garl\'ik and Ko{\l}odziejczyk~\cite{GK18SomeSubsystems}. Part and Tzameret~\cite{PT21Resolution} proved lower bounds for tree-like proofs over any field. They also obtain lower bounds for full, DAG-like proofs, but for instances that cannot be expressed as unsatisfiable CNF formulas.  Following a breakthrough of Efremenko, Garl\'ik, and Itsykson~\cite{EGI25LowerBounds}, CNF formula lower bounds and tradeoffs for \emph{regular} and \emph{bounded-depth} variants of resolution over $\F_2$-linear equations have recently been investigated by multiple authors~\cite{AI25Lifting, EI25Amortized, BBCI26LowerBounds, EI26StrongETH, IK26SupercriticalTradeoff, AN26NewPolynomial}.

Pang's paper~\cite{Pang26SizeDegree} mainly concerns \emph{resolution over polynomials}, which generalizes its counterpart over linear equations by allowing higher polynomial degree than $1$ in each disjunct. This proof system was introduced by Kraj\'{i}\v{c}ek~\cite{Krajicek18Randomized} and subsequently studied by Khaniki~\cite{Khaniki22ProofComplexity}, who established almost-quadratic lower bounds through a size-degree tradeoff different from Pang's.

\subsection{Paper Outline}\label{sec:paper-outline}
Following preliminaries in \cref{sec:preliminaries}, we prove \cref{thm:main} in \cref{sec:criterion} and \cref{thm:upper-bound} in \cref{sec:upper-bound}. We provide some concluding remarks in \cref{sec:concluding-remarks}. In Appendix~\ref{sec:sd}, we provide a straightforward extension of the results in \cite[Section~8]{Pang26SizeDegree} to arbitrary sets of inference rules for completeness and the convenience of the reader.

\section{Preliminaries}\label{sec:preliminaries}

We assume that clauses contain distinct variables and are non-tautological. The \emph{width} of a clause is the number of variables in it. For a
clause set $A$, let $N(A)$ be the union of its variable sets, identified
with a subset of $[n]$. We use the convention that Boolean value $1$ is
true. The polynomial encoding of a clause $C$ is
\begin{equation}\label{eq:clause}
p_C(x)=\prod_{x_i\in C}(1-x_i)
          \prod_{\neg x_i\in C}x_i\eqperiod
\end{equation}
Thus $p_C$ is $1$ precisely on assignments falsifying $C$ and is $0$
on the other Boolean assignments. Given a field $\F$ and a set of polynomials $P = \{p_1, \ldots, p_m\} \subseteq \F[x_1, \ldots, x_n]$, the set of polynomials of the form $\sum_i q_i p_i$ for arbitrary $q_i$ is called the \emph{ideal generated by $P$} and denoted $\ideal{P}$. We use
$\ln$ for natural logarithms and $\log$ for base-two logarithms.
All asymptotic statements concern $n\to\infty$.

\subsection{Algebra and Reduction Modulo Ideals}\label{sec:normal}

Fix a field $\F$ and write
$$
 \B\coloneqq \F[x_1,\ldots,x_n]/\langle x_i^2-x_i:i\in[n]\rangle.
$$
Each element of $\B$ has a unique multilinear representative. We
identify it with this representative and write $\ML(f)$ for the
image of an ordinary polynomial $f$ in $\B$. Ideals carry over from $\F[x_1, \ldots, x_n]$ to $\mathbb{B}$ by multilinearizing all polynomials in each ideal. Products in this section
and the next are products in $\B$, so they are implicitly
multilinearized. For a clause set $A$, define the ideal
$$
 J_A\coloneqq \langle p_C:C\in A\rangle\subseteq\B.
$$
An assignment satisfying $A$ defines an evaluation homomorphism
annihilating $J_A$, so $1\not\in J_A$ whenever $A$ is satisfiable.

Fix once and for all a degree-lexicographic order on the squarefree
monomials. For a nonzero polynomial $q\in\B$, let $\LM(q)$ be its
largest monomial in this order. A squarefree monomial is
\emph{reducible modulo} an ideal $J$ if it is $\LM(q)$ for some
nonzero $q\in J$; otherwise it is \emph{irreducible modulo}~$J$.

The rest of the section establishes some standard facts about reducibility modulo ideals. For proofs of the claims in~\cref{lem:normal}, see for instance~\cite[Section~2.2]{MN24Generalized}. 

\begin{lemma}\label{lem:normal}
For any ideal $J\subseteq\B$ and any $f\in\B$, there is a unique
linear combination $\redop_J(f)$ of irreducible monomials such that
$f-\redop_J(f)\in J$. The operator $f\mapsto \redop_J(f)$ is $\F$-linear and does not increase
degree. Moreover,
\begin{equation}\label{eq:normal-kernel}
\redop_J(f)=0\quad\Longleftrightarrow\quad f\in J.
\end{equation}
If $J\subseteq J'$ and $f,g\in\B$, then
\begin{equation}\label{eq:nested}
\redop_{J'}\bigl(g\,\redop_J(f)\bigr)=\redop_{J'}(gf).
\end{equation}
\end{lemma}

We abbreviate $\redop_{J_A}$ to $\redop_A$. 
Given $f \in \mathbb{B}$, we call the irreducible linear combination $R_J(f)$ in \cref{lem:normal} the \emph{reduction of $f$ modulo $J$} and the operator $R_J \colon f\mapsto R_J(f)$ the \emph{reduction operator modulo $J$}.

\begin{lemma}[Support locality]\label{lem:support}
For every clause set $A$ and polynomial $f\in\B$,
\begin{equation}\label{eq:support}
\Vars(\redop_A(f))\subseteq \Vars(f)\cup N(A).
\end{equation}
\end{lemma}

\begin{proof}
Put $S=\Vars(f)\cup N(A)$ and let $\B_S$ be the subalgebra on
variables in $S$. Let $J_A^S$ be the ideal of $\B_S$ generated by
the clause polynomials in $A$. Setting variables outside $S$ to zero
maps $J_A$ into $J_A^S$ and fixes all polynomials in $\B_S$.

A monomial $m$ supported on $S$ is reducible modulo $J_A$ if and only
if it is reducible modulo $J_A^S$. One implication follows from
$J_A^S\subseteq J_A$. For the other, take $q\in J_A$ with
$\LM(q)=m$ and set the outside variables to zero. This preserves
$m$ and deletes some other monomials, so the resulting element of
$J_A^S$ still has leading monomial $m$.

Consider the reduction of $f$ inside $\B_S$. It is congruent to $f$
modulo $J_A$ and uses only monomials irreducible modulo $J_A$.
Uniqueness in Lemma~\ref{lem:normal} therefore shows that it equals $\redop_A(f)$.
\end{proof}

\subsection{Polynomial Calculus}\label{sec:pc}

Given a field $\F$, \emph{polynomial calculus}~\cite{CEI96Groebner,ABRW02SpaceComplexity} is a proof system that proves that the polynomials in a set $\setP \subseteq \F[x_1, \ldots, x_n]$ have no common root by iteratively deriving new polynomials in the ideal $\ideal{\setP}$ until reaching the constant polynomial $1$. There are two derivation rules: 
  \begin{subequations}
    \begin{align}
        \label{eq:pc-lin-combn}
      \textit{Linear combination: }
      &\frac{\polyp \quad \polyq}{\alpha\polyp + \beta\polyq},\; \alpha, \beta \in \F \, ;\\
        \label{eq:pc-mult}
      \textit{Multiplication: }
      & \frac{\polyp}{x_i \polyp},\; x_i \text{ any variable.}
    \end{align}
    \end{subequations}
A \emph{(polynomial calculus) derivation} of a polynomial $\polyq$ from $\setP$ is a sequence $\refstd = (\polyp_1, \ldots, \polyp_\tau)$ where $\polyp_\tau = \polyq$ and each $\polyp_i$ in $\refstd$ is either in $\setP$ or obtained by applying one of \eqref{eq:pc-lin-combn}\nobreakdash--\eqref{eq:pc-mult} to polynomials $p_j, p_k \in \refstd$ with $j, k < i$. A \emph{refutation} of $\setP$ is a derivation of $1$. The \emph{degree} of a derivation is the maximum degree among the polynomials in it. We denote the minimum degree required to refute $\setP$ in polynomial calculus over $\F$ by $\DPC(\setP)$, and write $\setP \vdash_\degstd q$ if a polynomial $q$ can be derived from $\setP$ in degree at most $\degstd$. We write $\PC_{\degstd, P}$ to denote the set of polynomials derivable from $\setP$ in degree at most $\degstd$, and $\PC_{\degstd, P}(k)$ for $\PC_{\degstd, P}\intersection \F[x_1, \ldots, x_n]_{\leq k}$. We frequently drop the $\setP$ in the subscript when clear from context.

Polynomial calculus extends to $\mathbb{B}$ by representing all polynomials in a refutation modulo the Boolean axioms $\sdBool \coloneqq \{x_i^2 - x_i : i \in [n]\}$. It is clear that a degree lower bound for polynomial calculus over $\mathbb{B}$ extends to one over $\F$. 

Polynomial calculus naturally becomes a refutation system for unsatisfiable CNF formulas through the polynomial translation \eqref{eq:clause}. It is worth mentioning that in this setting, one usually introduces \emph{dual variables} $\overline{x}_i$ for each $x_i$ by including axioms $\overline{x}_i + x_i -1$ in $\setP$, but as we are only interested in degree lower bounds and each $\overline{x}_i$ in a refutation can be substituted by $1-x_i$ without increasing degree, we will not use this convention.

We frequently use the following basic facts of polynomial calculus derivations:
\begin{align}
 u\in\sdCD,\quad \deg u+\deg v\le D
 &\quad\Longrightarrow\quad uv\in\sdCD,
 \label{eq:sd-absorption}\\
 \sdBool&\vdash_{\deg u}u-\sdML(u).
 \label{eq:sd-multilinearization}
\end{align}
In particular, a polynomial vanishing on the Boolean cube has a
Boolean-axiom derivation in degree at most its own degree, and if a polynomial is derivable in degree $\degstd$ then so is its multilinearization. %
For an arbitrary degree-$\degstd$ polynomial $q$, we also have that
\begin{equation}
 \sdBool\vdash_{p\degstd}q^p-q
 \label{eq:sd-frobenius-identity}
\end{equation}
over $\F_p$. To see this, first note that
\begin{equation}
  \label{eq:frobenius-variable}
x_i^p-x_i
=(x_i^2-x_i)(1+x_i+\cdots+x_i^{p-2}) \eqperiod
\end{equation}
Over $\F_p$, we have for $q = \sum_j c_j m_j$ that 
\begin{equation}
  q^p - q = \sum_j c_j^p m_j^p - c_j m_j =  \sum_j c_j (m_j^p - m_j) \eqcomma
\end{equation}
and we can derive each $m_j^p - m_j$ from the Boolean axioms, one factor at a time, using \eqref{eq:frobenius-variable}. Equation~\eqref{eq:sd-frobenius-identity} follows.

Polynomial calculus degree lower bounds for refuting a set $\setP$ are most often proved by constructing a \emph{pseudo-reduction} operator for $\setP$. We define it over $\mathbb{B}$ for simplicity.
 \begin{definition}[Pseudo-reduction operator, \cite{Razborov98LowerBound}]
  \label{def:r-op}
  Let~$\degstd \in \N^+$ and let $\setP\subseteq \F[x_1, \ldots, x_\nvars]$ be unsatisfiable. An~$\F$-linear
  operator~$\pseudred$ on $\mathbb{B}$ 
  is a \emph{degree-$\degstd$ pseudo-reduction operator for~$\setP$} if
  \begin{enumerate}
  \item~$\pseudred(1) = 1$, \label{thesis:eq:rprop1}%
    \label{item:rop-property-1}
  \item~$\pseudred(\polyp) = 0$ for every polynomial~$\polyp\in\setP$,
    \label{item:rop-property-2}
  \item
  ~$\pseudred(x_im) =
    \pseudred\bigl(x_i
    \pseudred(m)\bigr)$ for every multilinear monomial~$m$ of degree at most
  $\degstd-1$ and every variable~$x_i$.
    \label{item:rop-property-3}
  \end{enumerate}
\end{definition}
The kernel of a pseudo-reduction operator overapproximates what can be derived from $\setP$ in degree at most~$\degstd$, but does not contain~$1$. It follows from a straightforward induction on refutation length that the existence of a degree-$\degstd$ pseudo-reduction operator for a set~$\setP$ implies a polynomial calculus degree lower bound of $\degstd+1$ for refuting~$\setP$.

\subsection{Constant-Depth Frege with Modular Gates}
\label{sec:AC0pFrege}

Fix a prime~$p$. The connective $\sdMOD_{p,a}(A_1,\ldots,A_m)$ is true
precisely when the number of true inputs is congruent to $a$ modulo~$p$.
We use a syntactic Frege system as in \cite[Section~2.1]{Krajicek19ProofComplexity} in the language
$\{\bot,\top,\neg,\wedge,\vee,\sdMOD_{p,a}:a\in\sdFp\}$.
Its inference rules are any fixed finite, sound, and implicationally complete
collection $\rules$ of Frege rule schemes, which we fix going forward. 

The modular connectives are governed by the standard recursive counting
axioms specifying how modular gates change with additional inputs; see, for instance, \cite[Section~10.1]{Krajicek19ProofComplexity}. We specify their
presentation because the depth of a proof depends on this choice.
For a sequence~$\Gamma$ of formulas and a formula~$B$, put
\[
 U=\sdMOD_{p,a}(\Gamma,B),\qquad
 V=\sdMOD_{p,a}(\Gamma),\qquad
 W=\sdMOD_{p,a-1}(\Gamma),
\]
where residues are taken modulo~$p$. To keep the axioms low-depth, we supply the recursion
$U\leftrightarrow((V\wedge\neg B)\vee(W\wedge B))$ in its
constant-size clausal form:
\[
\begin{aligned}
 B\vee\neg U\vee V,\qquad &B\vee U\vee\neg V,\\
 \neg B\vee\neg U\vee W,\qquad &\neg B\vee U\vee\neg W.
\end{aligned}
\]
We also supply $\sdMOD_{p,0}(\emptyset)$ and
$\neg\sdMOD_{p,a}(\emptyset)$ for $a\ne0$.
Explicitly, proof lines are Boolean formulas over
\[
\{\bot,\top,\neg,\wedge,\vee,\mathrm{MOD}_{p,a}; a\in\mathbb F_p\} \eqcomma
\]
Instances of rules and axioms are obtained by substituting formulas for their schematic variables.

An $\sdAC$-Frege refutation of depth $\ell$ of a CNF formula~$F$ is a sequence of formulas
ending with~$\bot$, where each line is an input clause, an axiom instance,
or follows from earlier lines by a rule in $\rules$, and each line has depth at most $\ell$ as specified below. Its \emph{size} is the
total number of symbol occurrences in the sequence. 

We measure the depth of a proof line in the gate convention. Literals and constants have depth
zero; negation does not increase depth; and each modular gate increases
depth by one. Each maximal block of consecutive conjunctions, or of consecutive
disjunctions, is counted toward depth as one unbounded-fan-in gate. Such a gate has depth
one plus the maximum depth of its inputs. %

A size-$s$ circuit of depth $\ell$ can be unfolded into an equivalent depth-$\ell$ formula of size $s^{O(\ell)}$, which for constant $\ell$ is polynomial in $s$. Hence, all bounds we obtain hold up to polynomial factors also for the version of the proof system where the lines are circuits instead of formulas.

\section{Degree Lower Bounds Through Matching Deficiency}\label{sec:criterion}

The goal in this section is to construct a pseudo-reduction operator as in \cref{def:r-op} for random CNF formulas. The properties needed to do so are captured in \cref{thm:criterion}. In \cref{sec:random} we prove that these conditions are satisfied by random CNF formulas with high probability. 

\begin{theorem}\label{thm:criterion}
Let $F$ be a CNF formula on $n$ variables, let $2\le \degstd\le n$, and let
$s\ge1$ be an integer. Suppose all clauses in $F$ have width at most $\degstd$
and the following conditions hold:
\begin{enumerate}[label=\textup{(\roman*)},leftmargin=2.5em,itemsep=0.2em]
   
 \item Every subformula of $F$ with fewer than $s$ clauses is satisfiable.
  \label{cond:local-satisfiability}
 \item For each \label{cond:matching-surplus}clause set $A\subseteq F$ with $s\le |A|\le2s$,
 \begin{equation}\label{eq:surplus}
|N(A)|>\degstd+|A| \eqperiod
 \end{equation}

\end{enumerate}
Then $\DPC(F)>\degstd$ over every field $\F$.
\end{theorem}

The proof of \cref{thm:criterion} occupies the rest of this section. As noted in the introduction, we follow the approach of Alekhnovich and Razborov~\cite{AR03LowerBounds}, where a pseudo-reduction operator $\pseudred$ is defined by associating to each low-degree multilinear monomial $m$ a \emph{satisfiable} clause set, usually called the \emph{closure} of $m$ and denoted $\cl{m}$. We then define $\pseudred$ on monomials by $\pseudred \colon m \mapsto R_{\ideal{\cl{m}}}(m)$ and extend linearly to arbitrary polynomials. Since closures of low-degree monomials are satisfiable, it is immediate with this definition that $R(1) = 1$. The challenge is to establish Properties \ref{item:rop-property-2} and~\ref{item:rop-property-3}, which we do in \cref{sec:r-op-properties}. The closure we define is field-independent, and the lower bound we establish holds over every field.

\subsection{Matching Deficiency and Closure}

In this section, we define the closure of a variable set and establish its basic properties. We identify variables with their indices. For $T\subseteq[n]$ with $|T|\le \degstd$ and a set $A\subseteq F$ with $|A|\le 2s$, define the \emph{matching deficiency} of $A$ with respect to $T$ as
\begin{equation}
  \label{eq:deficiency}
h_T(A)=|A|-|N(A)\setminus T|\eqperiod
\end{equation}
We call sets of size at most $2s$ \emph{admissible}. 

\begin{lemma}\label{lem:closure}
Among the admissible maximizers of $h_T$, there is a unique inclusion-wise largest
one. It has size strictly smaller than $s$.
\end{lemma}

\begin{proof}
The maximum is at least zero because $h_T(\varnothing)=0$. If
$s\le |A|\le2s$, condition~\eqref{eq:surplus} implies
$$
 h_T(A)\le |A|-|N(A)|+|T|<0.
$$
Thus every maximizer has fewer than $s$ clauses.

By recalling that $\nbh{U \intersection V} \subseteq \nbh{U} \intersection \nbh{V}$ and $\nbh{U \union V} = \nbh{U} \union \nbh{V}$, it is straightforward to show that $h_T$ is \emph{supermodular}:
\begin{equation}\label{eq:supermodular}
h_T(A\cup B)+h_T(A\cap B)\ge h_T(A)+h_T(B)\eqperiod
\end{equation}
If $A$ and $B$ maximize $h_T$, their union has size (strictly) smaller than $2s$.
Both sets on the left in~\eqref{eq:supermodular} are therefore admissible, and each has value at
most the maximum. Hence, \eqref{eq:supermodular} forces both to
attain the maximum. In particular, $A\cup B$ is again a maximizer
and thus has size smaller than $s$. Repeatedly taking unions gives a largest
maximizer, which is necessarily unique.
\end{proof}

We call the maximizer in \cref{lem:closure} the \emph{closure} of $T$ and denote it by $\cl{T}$.

\begin{lemma}[Monotonicity]\label{lem:monotone}
If $T\subseteq T'$ and $|T'|\le \degstd$, then
$\cl{T}\subseteq\cl{T'}$.
\end{lemma}

\begin{proof}
Let $A=\cl{T}$ and $B=\cl{T'}$. The difference
$$
 h_{T'}(I)-h_T(I)=|N(I)\cap(T'\setminus T)|
$$
is nondecreasing under inclusion of $I$. Combining this with
supermodularity gives
\begin{equation}
    \label{eq:ineq-deficiency}
 h_T(A)+h_{T'}(B)
 \le h_T(A\cap B)+h_{T'}(A\cup B).
\end{equation}
The union is admissible because both closures have size less than $s$.
Maximality of $A$ and $B$ bounds the right side of \eqref{eq:ineq-deficiency} by the
left side. Equality follows, so $A\cup B$ is a maximizer for
$h_{T'}$. Since $B$ is the largest such maximizer, we conclude that $A\subseteq B$.
\end{proof}

\begin{lemma}[Absorption]\label{lem:absorption}
Let $A=\cl{T}$. If $|U|\le \degstd$ and $U\subseteq T\cup N(A)$,
then $\cl{U}\subseteq A$.
\end{lemma}

\begin{proof}
Put $B=\cl{U}$ and suppose $B\not\subseteq A$. The admissible set
$A\cup B$ is a proper superset of the largest maximizer $A$, so
\begin{equation}
\label{eq:strict-deficiency}
h_T(A\cup B)-h_T(A)<0.
\end{equation}
On the other hand, $U\cup N(A\cap B)\subseteq T\cup N(A)$, and
therefore
$$
\begin{aligned}
h_U(B)-h_U(A\cap B)
 &=|B\setminus A|-|N(B)\setminus(U\cup N(A\cap B))|\\
 &\le |B\setminus A|-|N(B)\setminus(T\cup N(A))|\\
 &=h_T(A\cup B)-h_T(A)\\
 &<0 \eqperiod
\end{aligned}
$$
This contradicts maximality of $B$ for $h_U$, since $A\cap B$ is
admissible.
\end{proof}

We now show the local extendability property of the closure described in \cref{sec:proof-overview} that lets us establish Properties~\ref{item:rop-property-2} and \ref{item:rop-property-3} of a pseudo-reduction operator. 

\begin{lemma}[Local extendability]\label{lem:matching}
Let $A=\cl{T}$ and suppose $A\subseteq B\subseteq F$ with
$|B|\le2s$. There is a partial Boolean assignment $\rho$ which
satisfies every clause in $B\setminus A$ and assigns no variable in
$T\cup N(A)$.
\end{lemma}

\begin{proof}
For every $L\subseteq B\setminus A$, maximality of $A$ gives
$$
 0\ge h_T(A\cup L)-h_T(A)
   =|L|-|N(L)\setminus(T\cup N(A))|\eqperiod
$$
Consequently, Hall's condition holds in the bipartite graph whose left side is the clauses of
$B\setminus A$ and the right side is the variables outside $T\cup N(A)$, with an edge between a clause and a variable if the variable or its negation is contained in the clause.
Choose a matching covering these clauses and set each matched variable
to satisfy the literal in its matched clause. The matched variables
 are distinct, so the assignments are consistent and together satisfy all clauses in $B\setminus A$. 
\end{proof}

\begin{lemma}[Compatibility of local reductions]\label{lem:compatibility}
Let $A=\cl{T}$ and suppose $A\subseteq B\subseteq F$ with
$|B|\le2s$. If $\Vars(f)\subseteq T\cup N(A)$, then
$
\redop_B(f)=\redop_A(f)
$.
\end{lemma}

\begin{proof}
Put $S=T\cup N(A)$ and take the assignment $\rho$ from
Lemma~\ref{lem:matching}. Since $\rho$ satisfies $B\setminus A$
and leaves all generators from $A$ unchanged, restriction by $\rho$ maps
$J_B$ into $J_A$.

Suppose a monomial $m$ is supported on $S$ and $q\in J_B$ has leading monomial
$m$. Then $\res q\rho\in J_A$ still has leading monomial $m$.
Indeed, $m$ is untouched, while every other monomial either is untouched
 and remains smaller than $m$, vanishes, or loses degree. 

It follows that reducibility modulo $J_B$ and modulo $J_A$ coincide
on monomials supported on $S$; the reverse implication uses
$J_A\subseteq J_B$. By Lemma~\ref{lem:support}, $\redop_A(f)$ is
supported on $S$. It is therefore also irreducible modulo $J_B$ and
is congruent to $f$ modulo $J_B$. Uniqueness of reductions establishes the claim.
\end{proof}

\subsection{The Degree Lower Bound}
\label{sec:r-op-properties}

Let us now prove the degree lower bound claimed in \cref{thm:criterion}. For every multilinear monomial $m$ with $\deg m\le \degstd$, define
\begin{equation}
\label{eq:R}
R(m)\coloneqq\redop_{\cl{\Vars(m)}}(m)
\end{equation}
and extend linearly to the degree-at-most-$\degstd$ part of $\B$. Extend linearly and arbitrarily to higher degrees. We now show that $\pseudred$ is a degree-$\degstd$ pseudo-reduction operator for $F$.

\paragraph{Property \ref{item:rop-property-1}: $R(1)=1$.}
By Lemma~\ref{lem:closure}, $\cl{\varnothing}$ has fewer than $s$
clauses. Condition~\ref{cond:local-satisfiability} therefore implies that
$1\not\in J_{\cl{\varnothing}}$. Hence, since $1$ is minimal in the degree-lexicographic order, $R(1)=1$.

\paragraph{Property \ref{item:rop-property-2}: clause axioms vanish.}
Fix $C\in F$, put $U=\Vars(C)$, and set $B=\cl{U}$.
The closure is defined because $|U|\le \degstd$. We have $C\in B$:
otherwise $B\cup\{C\}$ is admissible and
$
 h_U(B\cup\{C\})=h_U(B)+1,
$
contradicting maximality. For every monomial $t$ occurring in $p_C$,
monotonicity gives $\cl{\Vars(t)}\subseteq B$.
Lemma~\ref{lem:compatibility} applied with the smaller closure and
with $f=t$ yields $R(t)=\redop_B(t)$. Therefore
\begin{equation}
\label{eq:Raxiom}
R(p_C)=\redop_B(p_C)=0 \eqperiod
\end{equation}

\paragraph{Property \ref{item:rop-property-3}: multiplication is compatible with reduction.}
Let $m$ be multilinear with $\deg m<\degstd$, and let $x_i$ be a variable.
Set
$$
 T=\Vars(m),\quad T'=T\cup\{i\},\quad
 A=\cl{T},\quad B=\cl{T'}.
$$
Here $\Vars$ is identified with variable indices. %
We have $A\subseteq B$ by monotonicity.

Every monomial $u$ of $R(m)=\redop_A(m)$ has degree at most $\deg m$
and is supported on $T\cup N(A)$. Thus the squarefree monomial
$v=x_i u$ in $\B$ has degree at most $\degstd$ and is supported on
$T'\cup N(B)$. Absorption applied to $B=\cl{T'}$ gives
$
 \cl{\Vars(v)}\subseteq B.
$
Lemma~\ref{lem:compatibility} consequently gives $R(v)=\redop_B(v)$.
Applying this term by term, and then using~\eqref{eq:nested}, yields
\begin{equation}
\label{eq:multiplication}
\begin{aligned}
R\bigl(x_iR(m)\bigr)
 &=\redop_B\bigl(x_i\redop_A(m)\bigr)\\
 &=\redop_B(x_i m)\\
 &=R(x_i m) \eqperiod
\end{aligned}
\end{equation}
The last equality uses $\Vars(x_i m)=T'$.
Thus, we have shown that $R$ is a degree-$\degstd$ pseudo-reduction operator for $F$, and \cref{thm:criterion} follows.

\begin{remark}[Provenance of techniques]
    The key notions in this section appeared before in several works. Franco and Van Gelder~\cite{FVG03Perspective} seem to be the first to study matching deficiency in the context of random CNF formulas. Proof-complexity applications were investigated at length by Kullmann~\cite{Kullmann00Application, Kullmann11ConstraintSatisfaction}. Notably, Kullmann uses deficiency to prove resolution upper bounds, while we employ it for lower bounds. Variations of the observations in Lemmas~\ref{lem:closure}\nobreakdash--\ref{lem:absorption} are present already in \cite[Section~4]{Kullmann11ConstraintSatisfaction}. The construction of the closure as a union of bounded-size sets is also an established technique in proof complexity, present already in \cite{ABRW04Pseudorandom, AR03LowerBounds} and later works such as~\cite{MN24Generalized}, and in other contexts including lower bounds for sum of squares~\cite{KMOW2017Sumofsquares} and for other algorithms for constraint satisfaction~\cite{CNP24Hierarchies,CN25Hierarchies}. 
\end{remark}

\subsection{Matching Deficiency and Local Satisfiability of Random CNF Formulas}\label{sec:random}

For sufficiently large $n$, let
\begin{equation}\label{eq:parameters}
\omega_n(1) \leq k \leq O(\log{n})\eqcomma \qquad
 M=n2^k \eqcomma \qquad
 s=\frac{n}{\sqrt{k}} \eqcomma \qquad
 \degstd=\frac n2-\frac{4n}{\sqrt{k}}\eqperiod
\end{equation}
In particular, $2\le k\le \degstd$ and $2s<n$.
Sample $M$ independent uniform signed $k$-clauses with replacement by choosing a uniformly random $k$-element subset of $[n]$ and then choosing independent uniform signs for the picked variables. The goal in this section is to show that, with high probability, a randomly sampled $k$-CNF formula with these parameters is unsatisfiable and satisfies the criteria in \cref{thm:criterion}.

\begin{proposition} \label{prop:random-cnf-properties} With probability $1-o_n(1)$ and parameters as in \eqref{eq:parameters}, a random $k$-CNF formula on $n$ variables with $M$ sampled clauses is unsatisfiable and satisfies Conditions~\ref{cond:local-satisfiability} and~\ref{cond:matching-surplus} in \cref{thm:criterion}. 
\end{proposition}

At first sight, the parameter choices for $s$ and $\degstd$ in \eqref{eq:parameters} seem rather arbitrary, but they turn out to be close to best possible for our applications; we elaborate on this in \cref{rmk:parameter-choice}. The rest of this section establishes \cref{prop:random-cnf-properties}.

\begin{proof}[Proof of \cref{prop:random-cnf-properties}] To show unsatisfiability, observe that a fixed Boolean assignment satisfies a random $k$-clause with probability
$1-2^{-k}$. Independence and a union bound over all assignments imply
\begin{equation}
\label{eq:unsat}
\begin{aligned}
\Pr[F\text{ is satisfiable}]
 &\le 2^n(1-2^{-k})^M\\
 &\le \exp(n\ln2-M2^{-k})\\
 &=\exp\bigl(-(1-\ln2)n\bigr)\eqperiod
\end{aligned}
\end{equation}

Let us now establish Condition~\ref{cond:local-satisfiability}. We show that, with probability $1-o(1)$,
\begin{equation}
\label{eq:local-Hall}
|N(A)|>|A|\qquad
 \text{for every nonempty clause set $A$ with $|A|\le2s$}.
\end{equation}
This implies Hall's condition in the constraint-variable incidence graph for every
subformula of at most $2s$ clauses. A matching there
then yields a satisfying assignment by setting
each matched variable to satisfy its matched clause.

If $1\le j<k$, any $j$ clauses have at least $k>j$ neighboring
variables. For $k\le j\le2s$, failure of~\eqref{eq:local-Hall} for
$j$ clauses means that their supports all lie inside some
$j$-element variable set. Therefore
\begin{equation}\label{eq:small-union}
\begin{aligned}
\Pr[\exists A:\ |A|=j,\ |N(A)|\le j]
 &\le \binom Mj\binom nj
       \left(\frac{\binom jk}{\binom nk}\right)^j\\
 &\le
       \left[e^2 2^k\left(\frac jn\right)^{k-2}\right]^j.
\end{aligned}
\end{equation}
The last inequality follows from $\binom ab\le(ea/b)^b$ and
\begin{equation}
\label{eq:hypergeom}
\frac{\binom jk}{\binom nk}
 =\prod_{a=0}^{k-1}\frac{j-a}{n-a}
 \le\left(\frac jn\right)^k
 \qquad(k\le j\le n)\eqperiod
\end{equation}
Since $j/n\le2/\sqrt{k}$, the bracket in~\eqref{eq:small-union}
is at most
$$
 B_k=e^2 2^k\left(\frac2{\sqrt{k}}\right)^{k-2}.
$$
We have
$$
 \ln B_k
 =2+(2k-2)\ln2-\frac{k-2}{2}\ln k\le-k
$$
for all sufficiently large $k$. Consequently, a union bound gives
\begin{equation}
\label{eq:small-failure}
\Pr[\eqref{eq:local-Hall}\text{ fails}]
 \le\sum_{j=k}^{2s}e^{-kj}
 \le\frac{e^{-k^2}}{1-e^{-k}}=o(1)\eqperiod
\end{equation}
Thus, Condition \ref{cond:local-satisfiability} is established. 

Finally, we turn to Condition~\ref{cond:matching-surplus}. 
Fix $s\le j\le2s$ and put $u=\degstd+j$. By~\eqref{eq:parameters},
\begin{equation}
\label{eq:u-bound}
u\le\frac n2-\frac{2n}{\sqrt{k}}.
\end{equation}
For a fixed $u$-element variable set $U$, the probability that a
random clause is supported inside $U$ is
$$
 p_u=\frac{\binom uk}{\binom nk}.
$$
Equations~\eqref{eq:hypergeom} and~\eqref{eq:u-bound} give
\begin{equation}\label{eq:medium-p}
p_u\le 2^{-k}\left(1-\frac4{\sqrt{k}}\right)^k
     \le 2^{-k}e^{-4\sqrt{k}}.
\end{equation}
(Here $k$ is sufficiently large that $1-4/\sqrt{k}>0$.)

If a set of $j$ clauses has at most $\degstd+j$ neighboring variables,
its supports all lie inside some $u$-element set. A union bound,
using $\binom nu\le2^n$ and $M=n2^k$, therefore yields
\begin{equation}\label{eq:medium-union}
\begin{aligned}
&\Pr[\exists A:\ |A|=j,\ |N(A)|\le \degstd+j]\\
 &\qquad\le\binom nu\binom Mj p_u^j\\
 &\qquad\le
 \exp\left(n\ln2+j\left[\ln\frac{en}{j}-4\sqrt{k}\right]\right) \eqperiod
\end{aligned}
\end{equation}
For sufficiently large $n$,
$$
 j\ge s\ge\frac{n}{2\sqrt{k}},\qquad
 \ln\frac{en}{j}\le\ln(2e\sqrt{k})\le\sqrt{k} 
 \eqperiod
$$
The right-hand side of~\eqref{eq:medium-union} is thus at most
\begin{equation}\label{eq:medium-exponent}
\exp(n\ln2-3j\sqrt{k})
 \le\exp\bigl(-({\tfrac32}-\ln2)n\bigr).
\end{equation}
Summing over $s\le j\le2s$, of which there are at most $n$ possible
values, proves
\begin{equation}\label{eq:medium-failure}
\Pr[\exists A:\ s\le|A|\le2s,\ |N(A)|\le \degstd+|A|]
 \le n\exp\bigl(-({\tfrac32}-\ln2)n\bigr)=o(1) \eqperiod
\end{equation}
This verifies Condition~\ref{cond:matching-surplus} with high probability.

Let us now finish the proof. Combining~\eqref{eq:unsat}, \eqref{eq:small-failure}, and
\eqref{eq:medium-failure}, the total
failure probability is at most
\begin{equation}\label{eq:total-failure}
e^{-(1-\ln2)n}
 +\frac{e^{-k^2}}{1-e^{-k}}
 +n e^{-({\tfrac32}-\ln2)n}
 =o(1)\eqperiod
\end{equation}
On the complementary event, the sampled CNF is unsatisfiable and the two hypotheses of
Theorem~\ref{thm:criterion} hold. 
The
resulting formula after deleting repeated clauses has width exactly $k$ and at most
$
 M
$
clauses, as desired. \qedhere
\end{proof}

\begin{remark}[Parameter choice]\label{rmk:parameter-choice}
    Due to the upper bound in \cref{thm:upper-bound}, the degree lower bound cannot be pushed beyond $n/2 - n\ln(k)/2k +O(n/k)$ and thus has no chance to contribute to higher-depth $\ACzero[p]$-Frege lower bounds, but even reaching this threshold seems non-trivial to achieve. The key parameter is the denominator in $s$, which is what determines the depth of the circuits in the proof lines for which the size lower bound applies. The value $\sqrt{k}$ comes from \eqref{eq:medium-union}\nobreakdash--\eqref{eq:medium-exponent} and is best possible with these calculations. %
    To see this, write the desired degree threshold as
$$
\degstd=\frac n2-\Delta\eqperiod
$$
For a $j$-clause set, the largest
permitted neighborhood in a bad event has size
$$
u=\degstd+j=\frac n2-(\Delta-j)\eqperiod
$$
Provided $j<\Delta$, the probability that a random clause is supported
inside a fixed $u$-element variable set satisfies
$$
\begin{aligned}
p_u
=\frac{\binom uk}{\binom nk}
&\le
\left(\frac12-\frac{\Delta-j}{n}\right)^k\\
&=
2^{-k}\left(1-\frac{2(\Delta-j)}{n}\right)^k\\
&\le
2^{-k}\exp\left(-\frac{2k(\Delta-j)}{n}\right).
\end{aligned}
$$
(If $u<k$, the containment probability is zero and there is
nothing to prove.)
Since $M=n2^k$, the union bound in \eqref{eq:medium-union} in this more general form gives
$$
\begin{aligned}
&\Pr\bigl[\exists A:\ |A|=j,\ |N(A)|\le \degstd+j\bigr]\\
&\qquad\le
\binom nu\binom Mj p_u^j\\
&\qquad\le
2^n
\left(\frac{en2^k}{j}\right)^j
\left[
2^{-k}\exp\left(-\frac{2k(\Delta-j)}{n}\right)
\right]^j\\
&\qquad=
\exp\left(
n\ln2+j\ln\frac{en}{j}
-\frac{2k}{n}j(\Delta-j)
\right).
\end{aligned}
$$
The saving depends on
$
j(\Delta-j), 
$
and even at its optimal value $j=\Delta/2$ we have
$$
j(\Delta-j)
=
\frac{\Delta^2}{4}
-\left(j-\frac{\Delta}{2}\right)^2
\le \frac{\Delta^2}{4}\eqperiod
$$
Consequently, the magnitude of the negative term in the exponent
is at most
$
k\Delta^2/{2n}.
$
For this union-bound estimate to overcome the $n\ln2$ term, we
therefore need to pick $\Delta$ so that
$
k\Delta^2/n = \Omega(n),
$
and hence
$
\Delta = \Omega(n/\sqrt{k}).
$
\end{remark}

\section{Polynomial Calculus Degree Upper Bounds for Random CNF Formulas}
\label{sec:upper-bound}

In this section we prove \cref{thm:upper-bound}. The key idea is a deterministic
criterion for a degree upper bound: if the supports of the clauses of an
unsatisfiable formula are sufficiently widely distributed, then all
monomials of large enough degree can be reduced to smaller degree. One
additional degree suffices to turn these reductions into a refutation.
We then show that for a random CNF formula $F$, the clauses obtainable by a
single resolution step from $F$ satisfy this criterion with high probability. 

\begin{theorem}[Quantitative form of Theorem~1.3]
\label{ub:theorem}
Let $n$ be a sufficiently large integer, $\omega_n(1) \leq k \leq O(\log{n})$, and put
\begin{equation}
M=n2^k,
\qquad
r=\frac n2-\frac{n}{2k}\ln\frac{k}{32}\eqperiod
\label{ub:parameters}
\end{equation}
Let $F$ consist of $M$ independently sampled uniform signed $k$-clauses
on $n$ variables. With probability $1-o_n(1)$, the formula $F$ is
unsatisfiable and
\begin{equation}
\operatorname{Deg}_{\mathrm{PC},\mathbb F}(F)\le r+1
\label{ub:main-bound}
\end{equation}
over every field $\mathbb F$. In particular, with the same probability, the bound
\begin{equation}
\operatorname{Deg}_{\mathrm{PC},\mathbb F}(F)
\le \frac n2-\frac{n\ln k}{2k}+O\!\left(\frac nk\right) 
\label{ub:asymptotic-bound}
\end{equation}
holds over every field.
\end{theorem}

We use the notation and clause encoding from Section~2. Although it is
convenient to write multilinearized products in the Boolean quotient
$\mathbb B$, all degree bounds below hold for ordinary polynomial
 calculus with the explicit Boolean axioms $x_i^2-x_i$. %

\subsection{Degree Upper Bound from Support Covering}
\label{ub:covering-subsection}

For $T\subseteq[n]$, write $x_T=\prod_{i\in T}x_i$. We say that a
clause set $F$ \emph{covers the $r$-element variable sets} if every
$T\subseteq[n]$ with $|T|=r$ contains $\operatorname{Vars}(C)$ for some
clause $C\in F$. %

\begin{lemma}[Support covering]
\label{ub:covering}
Let $1\le r\le n$, and let $F$ be an unsatisfiable CNF formula on $n$
variables whose clauses have width at most $r$. If $F$ covers the
$r$-element variable sets, then
$
\operatorname{Deg}_{\mathrm{PC},\mathbb F}(F)\le r+1
$
over every field $\mathbb F$.
\end{lemma}

\begin{proof}
Fix a field $\mathbb F$ and let $\mathbb B_{\le r+1}$ denote the vector
space of multilinear polynomials in $\F[x_1, \ldots, x_n]$ of degree at most $r+1$, and let
$\mathcal V\subseteq\mathbb B_{\le r+1}$ be the vector space of polynomials
whose multilinear representatives are derivable from $F$ and the
Boolean axioms in degree at most $r+1$. 
It is not hard to see that
\begin{equation}
f\in\mathcal V,\quad \deg f\le r
\quad\Longrightarrow\quad
\mathrm{ML}(x_i f)\in\mathcal V
\quad\text{for every }i\in[n] \eqperiod
\label{ub:bounded-multiplication}
\end{equation}

Fix the degree-lexicographic order on squarefree monomials. Let $\mathcal W$ be the span of those monomials of degree at most $r+1$ that are not leading monomials of nonzero elements of $\mathcal V$. Gaussian elimination gives
\[
\mathbb B_{\le r+1}=\mathcal V\oplus\mathcal W \eqperiod
\]
Define $R(f)$ to be the unique element of $\mathcal W$ such that $f-R(f)\in\mathcal V$. Thus $R$ is the linear projection onto $\mathcal W$ along $\mathcal V$. Indeed, existence follows by successively eliminating leading monomials using elements of $\mathcal V$. For uniqueness, a nonzero difference between two such representatives would belong to $\mathcal V$, while its leading monomial would belong to the spanning set defining $\mathcal W$---a contradiction. Since the monomial order respects degree, elimination never increases degree.
Moreover, since $1$ is the smallest
monomial in the ordering, we have that if $1\notin\mathcal V$ then $R(1)=1$. 

We claim that under the covering hypothesis, the image $\mathcal W=\operatorname{im}R$ consists of
polynomials of degree at most $r-1$. To see this, let $T$ have size $r$
or $r+1$. In the second case, first choose an $r$-element subset of
$T$. The covering hypothesis supplies a clause $C\in F$ with
$\operatorname{Vars}(C)\subseteq T$. Multiplying $p_C$ by
$x_{T\setminus\operatorname{Vars}(C)}$ gives, in degree at most
$|T|\le r+1$, a polynomial with leading monomial $x_T$ and nonzero
leading coefficient. Thus every monomial of degree $r$ or $r+1$ is
reducible modulo $\mathcal V$, proving the claim.

Now, suppose toward a contradiction that $1\notin\mathcal V$. We show that the claim above then yields a satisfying assignment for $F$, which contradicts unsatisfiability. The satisfying assignment will be built using a common eigenvector of a set of operators $\{X_i\}_{i\in [n]}$ on $\mathcal{W}$ that we now define. (As we note in \cref{rmk:eigenvalue-trick}, this eigenvector argument can be avoided in this particular setting, but we include it since it makes the proof rather neat and could have further applications.)

For every $f\in\mathbb B_{\le r}$, the polynomial $f-R(f)$ belongs to
$\mathcal V$ by definition and has degree at most~$r$. By
\eqref{ub:bounded-multiplication},
\begin{equation}
R(x_i f)=R\bigl(x_i R(f)\bigr)
\qquad(\deg f\le r) \eqcomma
\label{ub:reduction-compatibility}
\end{equation}
where products inside $R$ are multilinearized. Define linear
operators on $\mathcal W$ by
$$
X_i w=R(x_i w)\qquad(w\in\mathcal W) \eqperiod
$$
Because $\deg w\le r-1$, both $x_iw$ and $x_jw$ have degree at most
$r$. Applying \eqref{ub:reduction-compatibility} therefore gives
\begin{equation}
\begin{aligned}
X_iX_jw&=R\bigl(x_iR(x_jw)\bigr)
        =R(x_ix_jw)=R(x_jx_iw)=X_jX_iw \eqcomma\\
X_i^2w&=R(x_i^2w)=R(x_iw)=X_iw \eqperiod
\end{aligned}
\label{ub:commuting-idempotents}
\end{equation}
Thus the $X_i$ are commuting projections on $\mathcal{W}$.

For a polynomial $f$, let $f(X_1, \ldots, X_n)$ be the operator obtained from substituting the variables $x_i$ with $X_i$. A straightforward induction using \eqref{ub:reduction-compatibility} and \eqref{ub:commuting-idempotents} shows that
\begin{equation}
f(X_1,\ldots,X_n)1=R(f)
\qquad\text{for every }f\in\mathbb B_{\le r} \eqperiod
\label{ub:cyclic-vector}
\end{equation}

For each clause $C\in F$, its polynomial has degree at most $r$ and
belongs to $\mathcal V$, so
$$
p_C(X_1,\ldots,X_n)1=R(p_C)=0\eqperiod
$$
Since the operators commute and $1$ generates $\mathcal W$, it follows
that $p_C(X_1,\ldots,X_n)$ is zero on all of $\mathcal W$. Finally, we note that $\mathcal W\ne\{0\}$ because $R(1)=1$. 

To finish the proof, we use the following linear-algebraic fact. 
\begin{fact}
  \label{fact:commuting-projections}
  A finite set of commuting projections on a vector space $V \neq \{0\}$ has a common eigenvector $v$ where the corresponding eigenvalue for each operator is either $0$ or $1$.
\end{fact}
\begin{proof}
For simplicity, we prove the claim for two projections $\{P, Q\}$; the proof idea generalizes readily. Assume that both $P$ and $Q$ have nontrivial image, as otherwise the assertion is trivial. The projection $P$ splits $V$ as $V = \ker{P} \oplus \im{P}$. Note that if $v\in \ker{P}$ then $PQv = QPv = 0$, while 
if $v = Pu \in \im{P}$ then $Qv = QPu = PQu \in \im{P}$. 
Therefore, for any nonzero $v \in \im{P}$, $Qv$ is either nonzero and hence a common eigenvector of $P$ and $Q$ with eigenvalue $1$; or $Qv=0$, in which case $v$ itself is the desired common eigenvector.
\end{proof}

Choose a common eigenvector $v\ne0$ of the $X_i$ as in \cref{fact:commuting-projections} and write $X_iv=a_iv$ with
$a_i\in\{0,1\}$. For every $C\in F$,
$$
0=p_C(X_1,\ldots,X_n)v=p_C(a_1,\ldots,a_n)v \eqperiod
$$
Hence $p_C(a_1,\ldots,a_n)=0$ for every clause $C$, contradicting
unsatisfiability of $F$. We conclude that $1\in\mathcal V$, which establishes the
claimed degree upper bound. %
\end{proof}
\begin{remark}[Provenance of techniques]\label{rmk:eigenvalue-trick}
  The overall ideas for obtaining a refutation given the covering property are very similar to those underpinning the algorithm in~\cite[Section~4]{CEI96Groebner} that finds a polynomial calculus refutation of size $n^{O(\degstd)}$ given a degree upper bound of $\degstd$. 
  
  The eigenvector argument based on commuting projections in the second half of the proof of \cref{ub:covering} seems to be well known in, e.g., computational algebra. Similar arguments appear there in, e.g., \cite{Lazard81Resolution, AS88Elimination, KMH89Solving, YNT92Solutions}; see also~\cite[Sections 3 and 4]{Mourrain99Criterion}. 
  A survey paper by Cox~\cite{Cox21Stickelberger} dedicated to the eigenvector argument contains more discussion and references and relates it to classical results from the late 1800s. %
  
  Once \eqref{ub:cyclic-vector} is established, we can also prove the theorem without the eigenvector argument by instead appealing to the completeness of polynomial calculus: fix a refutation $\refstd$ of $F$ and replace each polynomial $f$ in it by $f(X_1, \ldots, X_n)1$. At the final proof line, we obtain a contradiction from linearity, $R(1)=1$, and \eqref{ub:reduction-compatibility}. %
\end{remark}

\subsection{Support Covering by Resolvents}
\label{ub:resolvents-subsection}

It remains to establish the covering property for a set of low-degree consequences of random $k$-CNF formulas. For a formula $F$, let $\Resol(F)$ be the set of clauses
obtained by resolving all clause pairs in $F$ whose variable sets intersect
in exactly the resolution variable. In more detail, the two clauses have the
form $z\vee A$ and $\neg z\vee B$, where $A$ and $B$ have disjoint
variable sets of size $k-1$, neither containing $z$. Their resolvent
$A\vee B$ then has width exactly $2k-2$, and it is straightforward to carry out the corresponding derivation in polynomial calculus by paying one extra degree. Since we will use this basic fact frequently, we state it formally below.

\begin{fact}
\label{ub:resolution-degree}
For $k\ge2$, every clause polynomial associated with a clause of
$\Resol(F)$ is derivable from $F$ in polynomial calculus degree
at most $2k-1$ over every field.
\end{fact}

Let us now state the lemma we prove in the rest of this section. 

\begin{lemma}
\label{ub:random-covering}
For the parameters in \eqref{ub:parameters}, let $F$ consist of $M$
independent uniform signed $k$-clauses. For all sufficiently large $n$,
\begin{equation}
\Pr\bigl[\Resol(F)\text{ does not cover the $r$-element variable sets}\bigr]
\le 2^{-n} \eqperiod
\label{ub:covering-probability}
\end{equation}
\end{lemma}

\begin{proof}
Split the \(M\) independent clauses into two batches of \(m=M/2\) clauses. Fix an \(r\)-element variable set \(T\) and write \(q\coloneqq n-r\ge n/2\).
For any fixed variable \(z\notin T\), the probability that a sampled width-$k$ clause has \(z\) as its only variable outside \(T\), with either sign, is
\begin{equation}
p=\frac{\binom r{k-1}}{\binom nk}\eqperiod
\end{equation}
The parameters in \eqref{ub:parameters} give
\[
mp=k\left(\frac{2r}{n}\right)^{k-1}(1+o(1))
   =32+o(1)\ge24
\]
for sufficiently large \(n\).
Say that \(z\) is \emph{represented} if the first batch contains such a clause. With very high probability, at least half of the \(q\) outside variables are represented. Indeed, otherwise some set of \(\lceil q/2\rceil\) outside variables is entirely unrepresented. For any fixed such set, a sampled clause represents one of its variables with probability \(\lceil q/2\rceil p\), since these possibilities are disjoint. A union bound therefore gives
\[
\Pr[\text{fewer than }q/2\text{ variables are represented}]
 \le 2^q(1-qp/2)^m
 \le 2^q e^{-12q}.
\]
Now condition on the first batch representing at least \(q/2\) variables. For each represented \(z\), choose one clause and write it as
\[
\ell_z\vee A_z,
\]
where \(\ell_z\) is a literal on \(z\) and all variables of \(A_z\) lie in \(T\).
A fresh clause from the second batch matches this chosen clause if it has the opposite literal on \(z\), its other variables lie in \(T\), and those variables avoid \(A_z\). Its probability of doing so is
\[
\frac p2\,
\frac{\binom{r-k+1}{k-1}}{\binom r{k-1}}
\ge \frac p4.
\]
For the inequality, a uniform \((k-1)\)-subset of \(T\) intersects the fixed support of \(A_z\) with probability at most \((k-1)^2/r=o(1)\).
Matching possibilities for different \(z\) are disjoint, because a matching clause has exactly one variable outside \(T\). Thus each second-batch clause matches some chosen clause with probability at least \(qp/8\). By independence, the probability that none does is at most
\[
(1-qp/8)^m\le e^{-3q}.
\]
Every match produces a resolvent in \(\Resol(F)\) supported entirely inside \(T\).
Combining the two failure probabilities, we obtain
\[
\Pr[\text{no resolvent is supported inside }T]
 \le 2^q e^{-12q}+e^{-3q}
 \le 2e^{-3q}
 \le 2e^{-3n/2} \eqperiod
\]
Finally, a union bound over the at most \(2^n\) choices of \(T\) yields
\[
\Pr[\Resol(F)\text{ does not cover the }r\text{-element sets}]
 \le 2^{n+1}e^{-3n/2}
 \le 2^{-n}
\]
for sufficiently large \(n\).
\end{proof}

\subsection{Proof of the Upper Bound}

We now have all required components to prove \cref{thm:upper-bound}.

\begin{proof}[Proof of \cref{thm:upper-bound}]
A fixed Boolean assignment satisfies a random $k$-clause with
probability $1-2^{-k}$. Hence, as in \cref{sec:random},
\begin{equation}
\Pr[F\text{ is satisfiable}]
\le2^n(1-2^{-k})^M
\le\exp\bigl(-(1-\ln2)n\bigr) \eqperiod
\label{ub:unsatisfiability}
\end{equation}
Intersect the unsatisfiability event with the support-covering event
from \cref{ub:random-covering}, and set
$$
G=F\cup\Resol(F)\eqperiod
$$
Then $G$ is unsatisfiable, its clauses have width at most $2k-2\le r$,
and it covers every $r$-element variable set. By
Lemma~\ref{ub:covering}, $G$ has a polynomial calculus refutation of
degree at most $r+1$ over every field.

By \cref{ub:resolution-degree}, every additional clause polynomial
in $G\setminus F$ is derivable from $F$ in degree at most $2k-1$.
Prepending these derivations to a refutation of $G$ gives a refutation
of $F$ of degree at most
$$
\max\{r+1,2k-1\}=r+1.
$$
The resulting degree bound therefore holds with probability at least
\begin{equation}
1-\exp\bigl(-(1-\ln2)n\bigr)-2^{-n}
=1-o_n(1).
\label{ub:total-probability}
\end{equation}
Finally, the definition of $r$ gives
$$
\begin{aligned}
r+1
&\le\frac n2-\frac{n\ln k}{2k}
       +\frac{n\ln32}{2k}+2\\
&=\frac n2-\frac{n\ln k}{2k}+O\!\left(\frac nk\right),
\end{aligned}
$$
which proves \eqref{ub:asymptotic-bound} and hence Theorem~1.3.
\end{proof}

To conclude this section, we observe that the weaker upper bound~${r \leq n/2-\Omega(n/k)}$ can be obtained by considering the formula $F$ itself instead of its set of resolvents and taking a straightforward union bound. This bound still suffices to rule out using random $k$-CNF formulas for Pang's tradeoff for $\ACzero[p]$-Frege proofs of depth greater than $2$. We provide the stronger result in order to not give false hope if lower-order improvements in the tradeoff parameters appear in the future. It seems plausible that considering the set of higher-width resolution consequences as in, e.g., PPSZ-style SAT algorithms~\cite{PPSZ05Improved} can further improve the bound on~$r$, but due to the nearly matching lower bound in \cref{thm:main} we do not explore this direction here. %

\section{Concluding Remarks}
\label{sec:concluding-remarks}
We prove that, with high probability, polynomial calculus requires degree $n/2 - O(n/\sqrt{\log{n}})$ to refute random $(\log{n})$-CNF formulas, while they can be refuted in degree $n/2 - \Omega(n \log(\log{n})/\log{n})$.
Through Pang's tradeoff, this implies that depth-$2$ $\ACzero[p]$-Frege requires size $\exp(\Omega_p(\log^{1.5}{n}))$ to refute these formulas, but that they cannot yield lower bounds for higher depths without improving the tradeoff parameters.

 The most pressing open problem related to this work is, of course, to prove a polynomial calculus degree lower bound of $n/2-o(n/\log^{C}(n))$ for a $\poly(n)$-size, $\polylog(n)$-width formula and every fixed $C>0$, since this yields superpolynomial lower bounds for $\ACzero[p]$-Frege. However, the upper bound in \cref{thm:upper-bound} seems to necessitate substantially different ideas from here to prove such a bound. One source of such ideas could be Beck and Impagliazzo~\cite{BI13StrongETH} and follow-up work of Bonacina and Talebanfard~\cite{BT16Improving}, which establish much stronger bounds for resolution width---on the order of $(1-\epsilon_k)n$ for $k$-CNF formulas, where $\epsilon_k \to 0$ as $k\to \infty$. Their hard instances are a certain encoding of sparse expanding systems of linear equations over finite fields. It seems, however, that straightforward adaptations of their techniques together with the machinery developed for binomial ideals (which cover linear equations) in \cite{BGIP01LinearGaps,BI10Random} lose too much in parameters to attain the needed degree lower bound. 

A related question, raised also in \cite{Pang26SizeDegree}, is whether the parameters in Pang's result are tight. The arguments in \cref{subsec:sd-common-multiplier} establishing the key kernel-growth estimate underpinning the tradeoff seem difficult to generically improve in any quantitatively significant manner. As discussed in \cref{rmk:sd-refined-errors}, a formula-specific approach to improve parameters is instead to find a formula $F$ for which quotients of the form $h_\dthresh(k+\Delta)/h_\dthresh(k)$ grow slowly enough, where $h_\dthresh(k)$ is the dimension of the space of degree-at-most-$k$ polynomials modulo the degree-$\dthresh$ consequences of $F$. Despite some efforts, we have not managed to obtain any useful results this way, but it still seems worth considering.

\section*{Acknowledgements}
I am grateful to Yassine Ghannane and Shuo Pang for helpful discussions about this and related works. In particular, I thank Shuo for generously sharing an early version of~\cite{Pang26SizeDegree} and clarifying some points in \cite[Section~8]{Pang26SizeDegree}. Thanks also to Lund University for giving their employees access to ChatGPT-$6$ Pro upon request. Part of this work was carried out during a visit to KTH Royal Institute of Technology; their hospitality is gratefully acknowledged. This work was partially supported by the \emph{Wallenberg AI, Autonomous Systems and Software Program (WASP)}.

\appendix

\section{From \texorpdfstring{$\ACzero[p]$}{AC0[p]}-Frege Size to Polynomial Calculus Degree}
\label{sec:sd}

In this section, we extend the size-degree tradeoff in \cite[Section~8]{Pang26SizeDegree} to an arbitrary set of inference rules. This can be inferred from the proof in \cite{Pang26SizeDegree}, but we provide a full proof here for completeness and convenience for the reader. Along the way, we provide more direct proofs of some technical arguments in \cite[Section~3]{Pang26SizeDegree}. The overall proof structure is, however, exactly the same as Pang's and we make no claim of originality regarding this. 

The two concrete size-degree relations we obtain are the following. 

\begin{theorem}[Fixed-depth tradeoff]
\label{thm:sd-fixed-depth}
Fix a prime~$p$ and an integer $\ell\ge2$. Let $F_n$ be a family of
unsatisfiable CNF formulas on $n$ variables of width $w=w(n)$.
Suppose that, for some $g=g(n)\ge0$,
\begin{equation}
 \sdPC(F_n)>\frac n2-g,
 \qquad
 g+w=o\left(\frac{n}{(\log n)^{\ell-2}}\right).
 \label{eq:sd-fixed-depth-hypothesis}
\end{equation}
Then every depth-$\ell$ $\sdAC$-Frege refutation of~$F_n$ requires size
$n^{\omega_n(1)}$. %
\end{theorem}

We can obtain slightly better parameters at depth $2$. 
\begin{theorem}[Depth-$2$ tradeoff]
\label{thm:sd-depth-two}
Fix a prime~$p$ and a constant $a>0$. Let
$\omega_n(1)\le k\le O(\log n)$, and let $F$ be an unsatisfiable
$n$-variable CNF formula of width at most~$k$. If
\begin{equation}
 \sdPC(F)>\frac n2-a\frac n{\sqrt{k}}\eqcomma
 \label{eq:sd-depth-two-hypothesis}
\end{equation}
then every depth-two $\sdAC$-Frege refutation of~$F$ requires size
\begin{equation}
 S\ge \exp\bigl(\Omega_{p,a}(\sqrt{k}\log n)\bigr)\eqperiod
 \label{eq:sd-depth-two-bound}
\end{equation}
\end{theorem}
Theorems~\ref{thm:sd-fixed-depth} and \ref{thm:sd-depth-two} are straightforward corollaries of a more technical statement, which we record in \cref{thm:sd-tradeoff} after introducing some notation. Recall that
$\sdBool=\{x_j^2-x_j:j\in[n]\}$ and $\polyspace{ \dthresh} = \F_p[x_1,\ldots, x_n]_{\leq \dthresh}.$ For a CNF formula~$F$ and a fixed cutoff~$D$, write
$
 Q(F)=\{p_C:C\in F\}\cup\sdBool$.  
 Recall that $\PC_{\dthresh}$ denotes the space
of polynomials derivable from $Q(F)$ in polynomial calculus degree at
most~$D$. Let
\begin{equation}
 h(k)=\dim \polyspace{k}
 - \dim{\PC_{\dthresh}\cap \polyspace{k}}
 \qquad(0\le k\le D)\eqperiod
 \label{eq:sd-hk}
\end{equation}
The same cutoff~$D$ is used at every argument of~$h$. For integers $0\le j\le n$, let
\begin{equation}
 N_n(j)=\sum_{u=0}^{j}\binom nu,\quad
 \eta_k(d)=\left(1-\frac1p\right)\frac{N_n(k-d)}{N_n(k)}
 \quad(0\le d\le k \le n), \quad \rho_k(d) \coloneqq (1-\eta_k(d))\eqperiod
 \label{eq:sd-eta}
\end{equation}
At this point, these quantities are rather mysterious. Briefly, $h(k)$ and $N_n(j)$ measure dimensions of polynomial vector spaces involved in the proof, and $\rho$ is an ``error-reduction factor'' obtained by raising the degree of the Razborov--Smolensky polynomials in \cref{def:razborov-smolensky-approximation} that we use to approximate the lines in an $\ACzero[p]$-Frege proof. Their roles will become clear later on; see \eqref{eq:growth-bound}. 
\begin{theorem}[Size--degree relation]
\label{thm:sd-tradeoff}
Fix a prime $p$ and a set $\rules$ of Frege rules defining $\ACzero[p]$\nobreakdash-Frege. There is a constant $B= B_{\rules, p}$ with the following
property. Let $F$ be an unsatisfiable width-$w$ CNF formula on $n$ variables, and let $D\ge0$ be an integer such that $\sdPC(F)>D$.
Let $\ell\ge1$ and let $t_1,\ldots,t_\ell$ be positive integers. Define
\begin{equation}
 \Delta_0=1\eqcomma\qquad
 \Delta_i=(p-1)t_i\Delta_{i-1}\quad(1\le i\le\ell)\eqperiod
 \label{eq:sd-degrees}
\end{equation}
Let~$K$ be an integer satisfying
\begin{equation}
 \Delta_{\ell-1}\le K\le n\eqcomma\qquad
 K+\max\{w+\Delta_\ell,\,B\Delta_\ell\}\le D\eqperiod
 \label{eq:sd-budget}
\end{equation}
For $1\le i\le\ell$, let
\begin{equation}
\epsilon_i=
 \rho_K(\Delta_{i-1})^{t_i} \eqperiod
 \label{eq:sd-refined-errors}
\end{equation}
Every depth-$\ell$ $\sdAC$-Frege refutation of~$F$ satisfies
\begin{equation}
 \sum_{i=1}^{\ell}M_i\varepsilon_i\ge1 \eqcomma
 \label{eq:sd-hilbert-rank-sum}
\end{equation}
where $M_i$ is the number of distinct
subformulas of depth~$i$ whose output gate is a conjunction or disjunction. In particular, every such refutation of
size~$S$ satisfies
\begin{equation}
 \ln S\ge
 \min_{1\le i\le\ell}
 -t_i\ln(\rho_K(\Delta_{i-1}))%
 \eqperiod
 \label{eq:sd-size-bound}
\end{equation}
\end{theorem}

We derive Theorems~\ref{thm:sd-fixed-depth} and \ref{thm:sd-depth-two} from \cref{thm:sd-tradeoff} in
Section~\ref{subsec:sd-parameters}. Until then, the goal is to prove \cref{thm:sd-tradeoff}. A key role in the proof will be played by the classical notion of \emph{Razborov--Smolensky (RS) approximation} of Boolean circuits. 
\begin{definition}[RS approximation~\cite{Razborov87LowerBounds,Smolensky87Algebraic}]\label{def:razborov-smolensky-approximation}
Let $p$ be a prime and $C$ a depth-$\ell$ formula with $\AND, \OR, \NOT$, and $\MOD_{p, c}$ gates, where $c$ can be any element of $\F_p$. For each $i\in [\ell]$, fix a positive integer $t_i$. The \emph{RS approximation $\App_{C}$} of $C$ is the following probabilistic polynomial defined inductively, bottom-up.\footnote{We have to be a bit careful and retain the syntactic formula tree with binary $\AND$ and $\OR$ connectives in the definition, since we later want to consider the set of all RS approximators of all syntactic subformulas in an $\ACzero[p]$-Frege proof. However, we deal with maximal consecutive blocks of the same connective simultaneously as expected.} 
\begin{itemize}
\item Constants, input gates, and negation are exact:
\begin{equation}
 \App_\bot=0,\qquad \App_\top=1,\qquad \App_{x_j}=x_j,
 \qquad \App_{\neg A}=1-\App_A\eqperiod
 \label{eq:sd-exact-basic}
\end{equation}
\item Modular connectives are also exact: for $A = \MOD_{p, b}(A_1, \ldots, A_m)$, define
\begin{equation}
 \App_{A}
 =1-\Big(\sum_{j=1}^m \App_{A_j}-b\Big)^{p-1} \eqperiod
 \label{eq:sd-mod-approx}
\end{equation}
\item For a maximal block of $\OR$ connectives~$A=\bigvee_{j=1}^mA_j$ at depth $i$, if $m< t_i$ then define the exact evaluation $\App_{A} = 1-\prod_j(1-\App_{A_j})$ in terms of the $\App_{A_j}$'s. Otherwise, for $\nu \in [t_i]$ define $q_{A, \nu} = \sum_{j} \alpha_{j, \nu} \App_{A_j}$ for uniformly random $\alpha_{j, \nu} \in \F_p$ and
\begin{equation}
  \App_A = 1 - \prod_{\nu=1}^{t_i}(1-q_{A,\nu}^{p-1}) \eqperiod
 \label{eq:sd-gate-OR}
\end{equation}
\item Similarly to above, for a maximal block of $\AND$ connectives~$A=\bigwedge_{j=1}^mA_j$ at depth $i$, if $m< t_i$ then define the exact evaluation $\App_{A} = \prod_j\App_{A_j}$ in terms of the $\App_{A_j}$'s. Otherwise, for $\nu \in [t_i]$ define $r_{A, \nu} = \sum_{j} \alpha_{j, \nu}(1-\App_{A_j})$ for uniformly random $\alpha_{j, \nu} \in \F_p$ and
\begin{equation}
  \App_A = \prod_{\nu=1}^{t_i}(1-r_{A,\nu}^{p-1}) \eqperiod
 \label{eq:sd-gate-AND}
\end{equation}
\end{itemize}
The approximation of $C$ is the approximation of its output gate. Hence, the polynomial $\App_{C}$ is an approximation of $C$ of degree at most $(p-1)^{\ell} \prod_{i\in [\ell]} t_i$. 
\end{definition}

\subsection{Technical Overview}
\label{sec:sd-technical-overview}

With the key definition of RS approximation in hand, we start by a (rather detailed) overview of the proof of \cref{thm:sd-tradeoff}.

\paragraph{Step 1: Approximating the Frege proof.}  Given a formula $F$, we fix a small $\ACzero[p]$-Frege refutation $\refstd = (C_1, \ldots, C_\tau=\bot)$ of $F$. We enumerate all syntactic subformulas of the proof lines and define their RS approximators consistently across $\refstd$, ensuring that whenever a subformula appears in two distinct lines, both occurrences get the same approximator, and also that $\App_{\lnot{A}} = 1-\App_A$. We then replace each $C_i$ by $p_{C_i} \coloneqq 1-\App_{C_i}$ to obtain a sequence $p(\refstd) = (p_{C_1}, \ldots, p_{C_\tau} = 1)$ of low-degree polynomials that approximate (the negation of) each~$C_i$. 

Our goal is to use $p(\refstd)$ and the polynomial calculus degree lower bound for refuting $F$ to obtain a contradiction. The most straightforward way to do this is to turn $p(\refstd)$ into a polynomial calculus refutation of $F$. While we do not do exactly that, it is instructive to consider what would be needed. 

The main obstacle is that in each line $L$ of $\refstd$, the polynomial $\App_{L}$ is an approximation of~$L$, so logical inferences in $\refstd$ cannot be expected to be directly simulated in $p(\refstd)$. There is a useful way to measure the extent to which such an inference fails. %
Suppose that, using a rule $\Rule$ (binary, for simplicity), we infer in $\refstd$ that
\begin{equation}
  \label{eq:example-inference}
  \Rule:\frac{A \quad B}{C} \eqperiod
\end{equation}
Upon replacing each circuit by its approximator, we instead get the approximate inference 
\begin{equation}\label{eq:approximate-modus-ponens}
  \frac{\App_{A}\quad \App_{B}}{\App_C} \eqperiod
\end{equation}
When is \eqref{eq:approximate-modus-ponens} invalid? Precisely when the premises are true but the conclusion is false, so the failure indicator is the polynomial $E= \App_{A}\App_B(1-\App_C)$. Algebraically, we can express this approximate inference through noting that
\begin{equation}
1 = 1-\App_A + \App_A(1-\App_B) + \App_A\App_B
\end{equation}
and multiplying by $p_C = 1-\App_C$ to get
\begin{equation}
  \begin{aligned}
  \label{eq:algebraic-inference}
  p_C 
  &= (1-\App_C)(1-\App_A) + \App_A(1-\App_B)(1-\App_C) \\
  &\quad+ \App_A\App_B(1-\App_C) \\
  &= (1-\App_C)(1-\App_A) + \App_A(1-\App_B)(1-\App_C) + E \\
  &= p_Cp_A + (1-p_A)p_Bp_C + E
  \end{aligned}
\end{equation}
Hence, we can derive the first two terms to the right in \eqref{eq:algebraic-inference} in low-degree polynomial calculus from the premises $p_A, p_B$ in $p(\refstd)$. What is not clear at this point is how to derive the error $E$ in low degree, and this is the next point to deal with.

\paragraph{Step 2: Removing error terms through a common multiplier.} Deriving $E$ in low degree seems like too much to hope for. However, using the \emph{common-multiplier method} as employed in~\cite{Pang26SizeDegree}, we do not need to. Instead, if we simulate all inferences in $\refstd$ in the form~\eqref{eq:algebraic-inference} and $F$ requires (very) high polynomial calculus degree to refute, then since there are not too many error polynomials in total, it turns out that we can pick $\dthresh$ large enough that there exists a low-degree $f\not\in \PC_{\dthresh}$ such that, nevertheless, $fE \in \PC_{\dthresh}$ for each inference error appearing in the proof. Thus, multiplying each line in $p(\refstd)$ by $f$ and noting that each line in $p(\refstd)$ has low degree by construction, each $f$-multiplied identity~\eqref{eq:algebraic-inference} becomes a legitimate low-degree polynomial calculus derivation
\begin{equation}
  \label{eq:f-mult-algebraic-inference}
  fp_C = p_C \cdot fp_A + (1-p_A)p_C \cdot fp_B+ fE \eqperiod
\end{equation}
By filling in these derivations for each line in $\refstd$, the sequence $fp(\refstd)$ becomes a low-degree polynomial calculus derivation of $f$---a contradiction. The remaining question is now what property of the error polynomials lets us conclude that such an $f$ exists.

\paragraph{Step 3: Deriving a structured error set.} So far, we have used nothing about the RS approximations of each proof line except that they are low-degree and that $1$ is retained exactly by the approximation. %
To see what makes the inference errors special, let us consider a concrete example where the inference rule in~\eqref{eq:example-inference} is modus ponens and we have an inference
\begin{equation}
  \label{eq:example-inference-modus-ponens}
  \frac{A \quad B \coloneqq (A \to C)}{C} \eqperiod
\end{equation}
This example is syntactically lighter than the general case but captures all the main ideas. Writing $A \to C = \lnot{A} \lor C$, we see that if $\App_{\lnot{A}\lor C} = 1-(1-\App_{C})\App_{A}$, then even though $\App_{A}$ and $\App_{B}$ are not necessarily correct representations of the circuits $A$ and $B$, we have $E\equiv 0$ and the inference is valid. The case that can go wrong is when the final gate of $C$ is an $\OR$ gate, so that $C= \bigvee C_j$ and syntactically $A \to C = \lnot{A} \lor \bigvee_j C_j$. In this case, the RS approximation $\App_{A \to C}$ might be quite different from $1-(1-\App_{C})\App_{A}$, since it will be sampled according to~\eqref{eq:sd-gate-OR} using the inputs $1-\App_A$ and $\{\App_{C_j}\}$. 

Now suppose that $E = 1$ on some input, which is to say that $\App_{A} = 1$, $\App_{A\to C} = 1$, and $\App_{C} = 0$. By \eqref{eq:sd-gate-OR} and the fact that $\App_{\lnot{A}} = 1-\App_{A}$, we must have for some $j$ that $\App_{C_j} = 1$ on that input. But then, since $\App_C=0$, there must be an error in the approximator of $C$, and we know that it occurred when approximating the output gate. We can write this error as a polynomial $E_j = \App_{C_j}(1-\App_C)$, and our argument shows that $E=1 \implies \exists j: E_j = 1$; contrapositively, this shows that $E$ is in the ideal generated by the $E_j$'s. In fact, due to the structure of the RS approximation, we can draw the stronger conclusion that $E$ is derivable in low-degree polynomial calculus from the $E_j$'s and the Boolean axioms. To see this, recall from~\eqref{eq:sd-gate-OR} that if $A \to C$ has depth $i$, then
\begin{equation}
    \label{eq:error-poly-derivation}
  \App_{A\to C} = 1 - \prod_{\nu=1}^{t_i}(1-q_{A\to C,\nu}^{p-1})
\end{equation}
where each $q_{A\to C,\nu}$ is a random linear combination of the polynomials $\App_{\lnot{A}} = 1-\App_A$ and $\{\App_{C_j}\}$. Since the constant term in \eqref{eq:error-poly-derivation} vanishes, each term in $\App_{A\to C}$ (if there are any) is a product of a nonempty subset of \smash{$\{q_{A\to C,\nu}^{p-1}: \nu \in [t_i]\}$}, and thus we can write
\begin{equation}
  \label{eq:form-of-approximator}
  \App_{A\to C} = r_0 (1-\App_{A}) + \sum_j r_j \App_{C_j}
\end{equation}
where each $r_i$ has degree at most $\Delta_i$. Therefore, since $E = \App_{A}\App_{A\to C}(1-\App_C)$, we have
  \begin{equation}
    \label{eq:derivation-error-polynomial}
  E 
  = r_0 (1-\App_{A})\App_{A}(1-\App_C) + \sum_j r_j \App_{A} \App_{C_j}(1-\App_C) \eqperiod 
  \end{equation} 
The first term on the right in \eqref{eq:derivation-error-polynomial} is derivable in low degree from the Boolean axioms since $\App_{A}$ and $\App_B$ are low-degree and Boolean-valued, and the second is a low-degree polynomial in the $E_j$'s as desired.

This argument generalizes to arbitrary syntactic inference rules without any conceptual changes; see~\cref{subsec:sd-simulation}. For the final part of the proof, we want to use the structure of the error families $E_j$ to locate the common multiplier.

\paragraph{Step 4: Finding the common multiplier through the error sets.}
We have shown that each inference error $E$ can be derived in low degree from the family of polynomials $\{E_j = \App_{C_j}(1-\App_{C})\}$. This is a rather special family, since $(1-\App_{C}) = \prod_{\nu=1}^{t}(1-q_{\nu}^{p-1})$, where each $q_{\nu}$ is an independent random linear combination of the $\App_{C_j}$'s. We can therefore think of the set $\{E_j\}$ as constructed by starting from the $\App_{C_j}$ and then repeatedly sampling random linear combinations and multiplying all $\App_{C_j}$ by the resulting $\prod_{\nu=1}^{t}(1-q_{\nu}^{p-1})$. 

To abstract finding a common multiplier in the setting above, given a degree threshold $\dthresh$ and a collection $G = \{g_1, \ldots, g_s\}$ of polynomials of degree at most $d \coloneqq (D-k)/(1 + t(p-1))$, define a random polynomial $p_j = \prod_{\nu=1}^{j}(1-q_{\nu}^{p-1})$ where each $q_\nu$ is a random linear combination of the $g_i$'s. Set $p_0 = 1$. Recall that $\polyspace{{\dthresh}}$ denotes the vector space of polynomials of degree at most~$\dthresh$. For $1\leq j \leq t$, sample each $q_j$ conditioned on $p_{j-1}$ and define maps
\begin{equation}
    \mult_{j, G} \colon \polyspace{ k} \to (\polyspace{ \dthresh} / \PC_{ D})^s; \quad \mult_{j, G}(f) = (f p_j g_i)_{i=1}^s \eqperiod
\end{equation}
Then the kernel of each $\mult_{j, G}$ consists of those degree-$\leq k$ polynomials $f$ such that each $f p_j g_i$ is a degree-$\leq \dthresh$ consequence of the input axioms, and since $\PC_{ D}$ is closed under low-degree multiplication we have that $\ker \mult_{j, G} \subseteq \ker \mult_{j+1, G}$ for all $j$. We want to show that after $t$ iterations, there exists $f\not\in \PC_{ \dthresh}$ in the intersection of the kernels of $\mult_{t, G}$ for each error family $G$ appearing in the proof. Equivalently, we want to show that
\begin{equation}
\expectation \Big[\dim{\bigcap_G \ker{\mult_{t, G}} / (\polyspace{ k} \intersection \PC_{ \dthresh})}\Big] > 0 \eqcomma
\label{eq:prob-estimate}
\end{equation}
where the expectation is over all choices of $q$'s for all families. To show \eqref{eq:prob-estimate}, note that if there are $S$ families in total, then with $h(k)$ as in \eqref{eq:sd-hk}, rank-nullity yields
\begin{equation}
  \label{eq:sum-estimate}
\begin{aligned}
  &\dim\Big(\bigcap_G \ker{\mult_{t, G}} / (\polyspace{ k} \intersection \PC_{ \dthresh})\Big)%
  \geq
  \sum_G \dim{\ker{\mult_{t, G}} / (\polyspace{ k} \intersection \PC_{ \dthresh})} - (S-1) h(k) \eqperiod
\end{aligned}
\end{equation}
The key lemma in this step is the kernel-growth estimate
\begin{equation}
    \label{eq:growth-bound}
  \expectation_{q_1, \ldots q_t} [\dim{(\ker{\mult_{t, G} / (\polyspace{ k} \intersection \PC_{ \dthresh})})}] \geq (1-\rho_k(d)^t)h(k) \eqcomma 
\end{equation}
where $\rho_k(d)$ is from \eqref{eq:sd-eta}. Indeed, taking expectations in \eqref{eq:sum-estimate} and applying \eqref{eq:growth-bound}, we get
\begin{equation}
  \expectation \Big[\dim{\bigcap_G \ker{\mult_{t, G}} / (\polyspace{ k} \intersection \PC_{ \dthresh})}\Big] \geq (1-S\rho_k(d)^t)h(k)
\end{equation} 
and the proof is complete once we pick parameters so that $S\rho_k(d)^t < 1$. 

\paragraph{Proof Roadmap.} As we have seen above, the two technical lemmas we need to establish are the kernel-growth estimate in \eqref{eq:growth-bound} and that we can derive the inference errors from approximation errors. We establish these lemmas in Sections~\ref{subsec:sd-common-multiplier} and \ref{subsec:sd-simulation}, respectively. We then carry out Steps 1 to 4 with explicit parameters in Section~\ref{subsec:putting-things-together}.

\subsection{Kernel Growth of Error Sets}
\label{subsec:sd-common-multiplier}

Recall that $\sdCD$ denotes the set of polynomials derivable in degree at most $\dthresh$ and $\PC_{\dthresh}(k)=\sdCD\cap V_k$. Hence, the quantity~$h(k)$ from~\eqref{eq:sd-hk} equals $\dim(V_k/\PC_{\dthresh}(k))$.
For $u\in V_D$, write $[u]$ for its equivalence class in
$V_D/\sdCD$.

Fix a family $G=(g_1,\ldots,g_s)$ of degree-at-most-$d$
polynomials. For a polynomial $\sdpolyp$ of degree at most~$a$ and $a+k+d\le D$,
define
\begin{equation}
 \begin{aligned}
 \mult_{\sdpolyp,k}:\polyspace{k}&\to(\polyspace{\dthresh}/\sdCD)^s,
 &\mult_{\sdpolyp,k}(u)&=([\sdpolyp g_iu])_{i=1}^s,\\
 Z_\sdpolyp(k)&=\ker \mult_{\sdpolyp,k},
 &\rankshort_{\sdpolyp,k}&\coloneqq \sdRank \mult_{\sdpolyp,k}=\dim(V_k/\ker \mult_{\sdpolyp,k}).
 \end{aligned}
 \label{eq:sd-error-hilbert}
\end{equation}
By closure under multiplication, $ \PC_{\dthresh}(k)\subseteq \ker \mult_{\sdpolyp,k}$ and therefore
\begin{equation}
 \rankshort_{\sdpolyp, k}\le h(k)\eqperiod
 \label{eq:sd-error-ambient}
\end{equation}

Finally, we need an estimate on the growth rate of $\rankshort_{\sdpolyp, k}$ with respect to $k$, which follows from comparing binomial coefficients.
\begin{lemma}[\cite{Pang26SizeDegree}, Lemma~4.1]
\label{lem:sd-hilbert-growth}
If $a+K+d\le D$ and $K\le n$, then $\rankshort_{\sdpolyp, k}/N_n(k)$ is nonincreasing
for $0\le k\le K$. In particular, for $d\le K$,
\begin{equation}
 \rankshort_{\sdpolyp, K-d}\ge\frac{N_n(K-d)}{N_n(K)}\rankshort_{\sdpolyp, K} \eqperiod
 \label{eq:sd-hilbert-growth}
\end{equation}
Also $h(k)/N_n(k)$ is nonincreasing for $0\le k\le\min\{D,n\}$.
\end{lemma}

We now establish the kernel-growth estimate in \eqref{eq:growth-bound}. It turns out to be more convenient to express it in terms of the rank of the involved multiplication maps, but by rank-nullity we can easily convert between the two formulations. The arguments in Steps 1 and 2 below are different from those in \cite{Pang26SizeDegree}, but reach the same conclusion. %

\begin{lemma}[Kernel growth]\label{lem:sd-kernel-growth}\label{thm:sd-one-layer}
Let $G=(g_1,\ldots,g_s)\subseteq \polyspace{d}$, where $1\le d\le K\le n$,
and let $t\ge1$. Set
\[
 e=(p-1)d,\qquad \Delta=te,\qquad
 \sdpolyp_t=\prod_{j=1}^{t}(1-q_j^{p-1}),
 \qquad q_j=\sum_{i=1}^{s}\alpha_{j,i}g_i,
\]
where the coefficients $\alpha_{j,i}$ are independent and uniform in
$\sdFp$. If $K+\Delta+d\le \dthresh$, then
\begin{equation}\label{eq:sd-kernel-growth-short}
 \sdE[\rankshort_{\sdpolyp_t,K}]
 \le
    \rho_K(d)^{t}h(K) \eqperiod%
\end{equation}
Equivalently, the expected dimension of
$\ker\mult_{\sdpolyp_t,K}/\PC_{\dthresh}(K)$ is at least $h(K)$
minus the right-hand side.
\end{lemma}

\begin{proof}
We use $\mult_{\sdpolyp,k}$ and $\rankshort_{\sdpolyp,k}$ from
\eqref{eq:sd-error-hilbert}; in particular,
$\rankshort_{\sdpolyp,k}\le h(k)$ by \eqref{eq:sd-error-ambient}.
The proof proceeds in three steps. The first two concern general properties of kernels of multiplication maps as in \eqref{eq:sd-error-hilbert}. 

In Step~1, we study combinatorial lines in $V_d$ of the form $\ell = \{q + c r: c\in \F_p\}$ for polynomials $q, r$ and prove that the only elements in the shared kernels of the maps $\{\mult_{q+cr}: c \in \F_p\}$ are the obvious ones, namely those in $\ker{\mult_q}\intersection \ker{\mult_r}$. This lets us bound the sum $\sum_{c\in \F_p}{\dim{\ker{\mult_{q+cr}}}}$ in terms of the dimensions of the domain $U$, the set $\ker{\mult_r}$, and of the joint kernel $\ker{\mult_q}\intersection \ker{\mult_r}$. 

In Step~2, we interpret a random linear combination $\sum_{i \in [s]}\alpha_i g_i$ as a line with the first $s-1$ terms as the base and the last term as the direction. The estimates in Step $1$ then easily let us establish a lower bound on the expected rank of $\mult_{\Sigma_{i \in [s]}\alpha_i g_i}$ by induction on $s$. 

In the final step, we apply this bound to the multiplication maps of interest. Roughly speaking, the idea is that for a random linear combination $q=\sum_i \alpha_i g_i$, the expected rank of $\mult_{q}$ is noticeably large by the bound in Step~2. Hence, so is the set of low-degree $q$-multiples. When we multiply by $(1-q^{p-1})$, we send all low-degree $q$-multiples $qu$ to $0$ modulo $\PC_\dthresh$, since $q-q^p$ is derivable from the Boolean axioms. Thus, we are able to show that the kernel grows considerably when we do so.

\paragraph{Step 1: Separating kernels.}
In this step, we prove the following general claim. Let $U\subseteq\polyspace{k}$ be a linear subspace and let
$h_1,\ldots,h_m$ have degree at most $b$, where $b+k+pd\le\dthresh$.
For $g\in\polyspace{d}$, define
\[
 \mult_g\colon U\to(\polyspace{\dthresh}/\sdCD)^m,
 \qquad \mult_g(u)=([h_j gu])_{j=1}^m.
\]
Fix $q,r\in\polyspace{d}$, and put
\[
 L_c=\ker \mult_{q+cr}\quad(c\in\sdFp),\qquad
 L_\infty=\ker \mult_r,\qquad J=\ker \mult_q\cap\ker \mult_r.
\]
We claim that the spaces $L_c/J$, including $L_\infty/J$, have direct
sum in $U/J$. To show this, we define the polynomial
$\Psi_a(X,Y)=Y(1-(X+aY)^{p-1})$. We want to prove that $\Psi_a$ isolates the summand indexed by $a$.
Indeed, writing $\delta_{a,c}=\mathbbm{1}_{a=c}$, we have
\begin{equation}\label{eq:sd-kernel-separator}
 \Psi_a(X,Y)-\delta_{a,c}Y
 =(X+cY)A_{a,c}(X,Y)+(1-\delta_{a,c})(Y-Y^p),
\end{equation}
for some polynomial $A_{a, c}$ with $\deg A_{a,c}\le p-1$.
To check divisibility by $X+cY$, substitute $X=-cY$ and use
$(a-c)^{p-1}=1-\delta_{a,c}$.
For $u_c\in L_c$, substitute $X=q$, $Y=r$ and multiply by
$h_\alpha u_c$. The first term on the right in \eqref{eq:sd-kernel-separator} is then $h_\alpha (q+cr)u_c A_{a, c}(q, r)$, which belongs to $\sdCD$ by definition of $L_c$ and closure of $\sdCD$ under low-degree multiplication. The second term on the right is also in $\sdCD$ since
$r^p-r$ is derivable from the Boolean axioms. Thus
\begin{equation}\label{eq:sd-kernel-separation}
 [h_j\Psi_a(q,r)u_c]=\delta_{a,c}[h_j ru_c] \eqperiod
\end{equation}
Every polynomial involved in this argument has degree at most $b+k+pd$, which is at most $\dthresh$ by assumption. Note also that the argument covers $c=0$. 

Now let $u_c\in L_c$ for $c\in \F_p$, and
$\sum_{c\in\sdFp}u_c+u_\infty\in J$.
Multiply this relation by $h_\alpha\Psi_a(q,r)$ and consider the equivalence class modulo $\sdCD$.
The $u_\infty$ and $J$ terms vanish because $\Psi_a(q,r)$ is divisible
by $r$, while \eqref{eq:sd-kernel-separation} isolates $h_\alpha ru_a$.
Hence $\mult_r(u_a)=0$. Together with the assumption that $\mult_{q+ar}(u_a)=0$, this gives
$u_a\in J$. This holds for every $a$, and the original relation then
gives that $u_\infty\in J$. The direct-sum claim follows. Taking dimensions, we conclude in particular that
\begin{equation}\label{eq:sd-kernel-pencil-dimension}
 \sum_{c\in\sdFp}\dim L_c
 \le \dim U-\dim L_\infty+p\dim J\eqcomma
\end{equation}
which is the estimate we need for the next, inductive step.

\paragraph{Step 2: Averaging over a random linear combination.}
With the same maps and degree assumptions, we claim that
\begin{equation}\label{eq:sd-rank-averaging-short}
 \sdE_{\alpha}\sdRank \mult_{\sum_i\alpha_i g_i}
 \ge\left(1-\frac1p\right)
 \dim\Big(U/\bigcap_i\ker \mult_{g_i}\Big).
\end{equation}
Let $Z=\bigcap_i\ker \mult_{g_i}$. By rank-nullity, it suffices to bound the expected kernel dimension from above
by $\dim U/p+(1-1/p)\dim Z$. We prove the claim by induction on the number $s$ of polynomials, where the case $s=0$ is trivial. For the induction step, put
\[
 r=g_s\eqcomma\qquad q=\sum_{i<s}\alpha_i g_i\eqcomma\qquad U_0=\ker \mult_r\eqperiod
\]
For each fixed $q$, divide \eqref{eq:sd-kernel-pencil-dimension} by $p$ to get
\begin{equation}
 \sdE_{c\in\sdFp}\dim\ker \mult_{q+cr}
 \le\frac{\dim U-\dim U_0}{p}
       +\dim(U_0\cap\ker \mult_q)\eqperiod
    \label{eq:upper-bound-induction-step}
\end{equation}
The first $s-1$ maps restricted to $U_0$ have common kernel $Z$.
The induction hypothesis applied to this subspace gives
\begin{equation}
 \sdE_q\dim(U_0\cap\ker \mult_q)
 \le\frac{\dim U_0}{p}+\left(1-\frac1p\right)\dim Z \eqperiod
 \label{eq:induction-hypothesis}
\end{equation}
Substituting \eqref{eq:induction-hypothesis} into \eqref{eq:upper-bound-induction-step} cancels the $\dim U_0/p$ terms and
proves \eqref{eq:sd-rank-averaging-short}.

\paragraph{Step 3: The final bound.}
Let $\deg\sdpolyp\le a$, suppose $a+K+pd\le \dthresh$,
and let $q$ be a random linear combination
of $G$. Apply \eqref{eq:sd-rank-averaging-short} on $\polyspace{K-d}$ with
$h_\alpha=\sdpolyp g_\alpha$.
The degree requirement to do so is $(a+d)+(K-d)+pd\le\dthresh$, which is satisfied.

We now want a more transparent right side in this application of \eqref{eq:sd-rank-averaging-short}. For this, we show that the common kernel of the maps
$u\mapsto([\sdpolyp g_j g_i u])_j$ is exactly
$\ker\mult_{\sdpolyp,K-d}$. One inclusion follows by multiplying its
defining relations by $g_i$. Conversely, consider $u$ in the common kernel. For $j = i$ we then have in particular that
$\sdpolyp g_i^2u\in\sdCD$. Multiply by $g_i^{p-2}$ to get
$\sdpolyp g_i^pu\in\sdCD$. Since we can derive $g_i^p-g_i$ over $\F_p$ from the Boolean axioms by \eqref{eq:sd-frobenius-identity}, it follows that
$\sdpolyp g_iu\in\sdCD$ for all $i$. All these operations fit the degree budget.
Consequently, rank-nullity and \eqref{eq:sd-rank-averaging-short} give
\begin{equation}\label{eq:sd-selected-rank-short}
 \sdE_q[\sdRank\bigl(u\mapsto([\sdpolyp g_j qu])_j\bigr)]
 \ge\left(1-\frac1p\right)\rankshort_{\sdpolyp,K-d}\eqperiod
\end{equation}

Put $\sdpolyp'=\sdpolyp(1-q^{p-1})$. The kernel then increases as follows:
\begin{equation}\label{eq:sd-kernel-enlargement-short}
 \ker\mult_{\sdpolyp,K}+q\polyspace{K-d}
 \subseteq\ker\mult_{\sdpolyp',K}.
\end{equation}
The inclusion of the first summand follows from closure of $\PC_\dthresh$ under low-degree multiplication. For the second, note that
$\sdpolyp'g_i(qu)=\sdpolyp g_iu(q-q^p)$, so the claim follows from \eqref{eq:sd-frobenius-identity} since $\PC_D$ contains the Boolean axioms.
Both calculations have degree at most $a+K+pd$.
Taking dimensions in \eqref{eq:sd-kernel-enlargement-short} gives
\begin{equation}
  \label{eq:weird-map}
 \rankshort_{\sdpolyp',K}\le\rankshort_{\sdpolyp,K}
 -\dim\frac{\ker\mult_{\sdpolyp,K}+q\polyspace{K-d}}
                 {\ker\mult_{\sdpolyp,K}} \eqperiod
\end{equation}
The subtracted dimension on the right in \eqref{eq:weird-map} is the rank of the map
$T_q: u\mapsto qu+\ker\mult_{\sdpolyp,K}$ on $\polyspace{K-d}$. We now show that this quantity is precisely the rank of the map in the left side of~\eqref{eq:sd-selected-rank-short}. It suffices to show that the two maps have the same kernel; the claim then follows from rank-nullity since they also have the same domain. Indeed,
\[
T_q(u)=0
\iff qu\in\ker\mult_{\sdpolyp,K}
\iff [\sdpolyp g_j qu]=0\quad\text{for every }j \eqcomma
\]
as desired. 
Taking expectations and using \eqref{eq:sd-hilbert-growth} therefore gives
\begin{equation}\label{eq:sd-contraction-short}
 \sdE_q [\rankshort_{\sdpolyp',K}]
 \le \rankshort_{\sdpolyp,K}-\left(1-\frac1p\right)\rankshort_{\sdpolyp,K-d}
 \le\rho_K(d)\rankshort_{\sdpolyp,K} \eqcomma
\end{equation}
where $\rho_K(d)$ is as in %
\eqref{eq:sd-eta}.

To finish the proof, write $\sdpolyp_0=1$ and
$\sdpolyp_j=\prod_{\nu=1}^j(1-q_\nu^{p-1})$, so
$\deg\sdpolyp_j\le je$. Condition on $q_1,\ldots,q_{j-1}$ and apply
\eqref{eq:sd-contraction-short}:
\begin{equation}
 \sdE[\rankshort_{\sdpolyp_j,K}\mid q_1,\ldots,q_{j-1}]
 \le\rho_K(d)\rankshort_{\sdpolyp_{j-1},K} \eqperiod
\end{equation}
The degree requirement to apply \eqref{eq:sd-contraction-short} is $K+(j-1)e+pd\le K+(t-1)e +pd = K+\Delta+d\le\dthresh$, which is satisfied by assumption.
Taking expectations and iterating gives
$\sdE[\rankshort_{\sdpolyp_t,K}]\le\rho_K(d)^t\rankshort_{1,K}
\le\rho_K(d)^t h(K)$, as desired.\qedhere

\end{proof}

\begin{remark}[Potential improvement in \eqref{eq:sd-kernel-growth-short}]
  \label{rmk:sd-refined-errors}
  The factor $\rho_K(d)^{t}$ in~\eqref{eq:sd-kernel-growth-short} is the same as in the
size--degree relation in \cite[Section~8]{Pang26SizeDegree}. Under the slightly stronger assumption \(K+\Delta+pd\le D\), a variation of the argument for obtaining~\eqref{eq:sd-contraction-short} can replace the final estimate by
\begin{equation}
  \expectation [\rankshort_{\sdpolyp_t, K}] \leq p^{-t} h(K+\Delta)
\end{equation}
which can improve the estimates when the actual function~$h$ associated to $F$ grows slowly. However, it seems challenging to estimate $h$ well enough for the difference to matter. The only polynomial calculus degree lower bound this author is aware of from which one can read off an exact value for $h$ is Razborov's lower bound for variants of the pigeonhole principle~\cite{Razborov98LowerBound,IPS99LowerBounds}. Unfortunately, the parameters there seem insufficient for a useful improvement over $\rho_K(d)^{t}h(K)$ (which is also insufficient to extract a meaningful bound from the tradeoff in this case). 
\end{remark}

\subsection{Deriving Inference Errors from Approximation Errors}
\label{subsec:sd-simulation}
The next technical result we need is to derive the inference errors in the proof overview from errors in the RS approximators of the proof lines. The proof is very similar to its counterpart in \cite{Pang26SizeDegree} and builds on ideas from~\cite[Section~6]{BIKPRS97ProofCplx}. The formal statement is \cref{lem:sd-inference-errors}. For every
conjunction or disjunction $A$ with inputs $A_1,\ldots,A_m$, we define the shorthands
\begin{equation}
    \label{eq:shorthands}
 (g_{A,j},R_A)=
 \begin{cases}
 (\App_{A_j},1-\App_A),&A=\bigvee_jA_j,\\
 (1-\App_{A_j},\App_A),&A=\bigwedge_jA_j,
 \end{cases}
 \qquad e_{A,j}=g_{A,j}R_A.
\end{equation}
\begin{lemma}
\label{lem:sd-inference-errors}
Fix a set $\rules$ of Boolean inference rules and modular recursion axioms defining
$\sdAC$\nobreakdash-Frege. There is a constant $B = B(\rules, p)$ with the following
property. Assign consistent RS approximators $\App_A$ to all syntactic
rooted subformulas of an $\ACzero[p]$-Frege proof, with $\deg \App_A\le\Delta$. 
For each non-input inference
$L_1,\ldots,L_r \vdash L$ (including zero-premise axioms), its inference error
polynomial
\[
 E=\Bigl(\prod_{j=1}^{r}\App_{L_j}\Bigr)(1-\App_L)
\]
has a polynomial calculus derivation from the $e_{A,j}$ and the Boolean
axioms in degree at most~$B\Delta$. 

Moreover, if the input formula of the proof has width at most $w$, then for each input clause $C$ its approximation $1-\App_C$ is derivable from the polynomial representation of $C$ and its approximation errors in degree at most $w + \Delta$. 
\end{lemma}

\begin{proof}
The proof is conceptually the same as for the modus ponens example in Step~3 of the proof overview: by soundness of the inference rule, an inference error must be caused by an error in approximating the premises or conclusion, and there is in fact a low-degree algebraic derivation of the inference error from the approximation errors. However, the proof is a bit more cumbersome than there since we need some algebraic manipulations to express the corresponding notions as for~\eqref{eq:example-inference-modus-ponens} for a general inference rule. We deal with the three kinds of inferences in order: applications of Frege rules; modular recursion axioms; and finally input inferences. 

\paragraph{Frege rules.} Suppose that the inference $L_1,\ldots,L_r \vdash L$ was made using the Frege rule 
\begin{equation}
  \label{eq:rule-formal}
  \Rule \colon \frac{F_1, \ldots, F_r}{F} \eqperiod
\end{equation}
Here, as in \eqref{eq:example-inference-modus-ponens}, each $F_i$ and $F$ is a propositional formula in the de Morgan language expressed using a shared set of
Boolean metavariables, which, we emphasize, only depend on $\Rule$ and in particular are not the variables of the $L_i$. The inference $L_1, \ldots, L_r \vdash L$ is obtained by a substitution $\sigma(F_i)=L_i, \sigma(F)=L$, where each metavariable is assigned the corresponding syntactic subformula of the $L_i$ or $L$ by~$\sigma$. 

For each metavariable $X$ of the $F_i$ and $F$, we assign $a_X = \App_{\sigma(X)}$. For each subformula $G$ of the $F_i$ or $F$, let $G(a)$ be its exact
Boolean evaluation on the $a_X$'s, expressed as a polynomial using that
\[
 u\wedge v=uv,\qquad
 u\vee v=u+v-uv,\qquad
 \neg u=1-u.
\]
Our goal is now to derive the polynomial
\begin{equation}
    \label{eq:approximation-errors}
 D_G=\App_{\sigma(G)}-G(a)
\end{equation}
from approximation errors and the Boolean axioms, working upward through $G$. We then derive $E$ from $D_{F_i}$ and $D_F$ by using that the polynomial 
\begin{equation}
    \label{eq:sound-rule-polynomial}
 E^*=\Big(\prod_{i=1}^rF_i(a)\Big)(1-F(a))
\end{equation}
vanishes on the Boolean cube by soundness of $\Rule$. 

Going forward, we often use the fact that any degree-$d$ polynomial vanishing on the
Boolean cube is derivable from the Boolean axioms in degree $d$. In particular, every RS approximator is Boolean-valued
on the Boolean cube, so $\App_A(1-\App_A)$ is derivable in degree at most
$2\Delta$ for every subformula $A$.

We first show how to proceed upwards toward \eqref{eq:approximation-errors} in a subformula through a single binary connective. For this, we use the identity 
\begin{equation}
 1-R_A=\sum_j r_{A,j}g_{A,j}\eqcomma \qquad \text{for some $r_{A, j}$ with }\deg r_{A,j}\le\Delta \eqperiod
 \label{eq:sd-short-gate-certificate} 
\end{equation}
which holds in the randomized cases of \eqref{eq:sd-gate-OR}\nobreakdash--\eqref{eq:sd-gate-AND} by the same argument as in \eqref{eq:form-of-approximator}, while the same identity holds in the exact cases by telescoping. Now, consider syntactic formulas $A,B$. We first consider $A\lor B$, where we want to derive
\begin{equation}
    \label{eq:difference-eq-OR}
    \App_{A\lor B} - (\App_A + \App_B -\App_A\App_B) \eqperiod
\end{equation}
(Note that this polynomial can be nonzero since we group maximal blocks of the same connective when approximating $\OR$ and $\AND$ gates.) Write 
\[
 u=\App_A,\qquad v=\App_B,\qquad z=\App_{A\vee B}.
\]
If $A = \bigvee_{i\in I} A_i$, then let $g_i$, $i\in I$ be the approximators of each $A_i$. If $A$ does not start with OR, use the singleton list consisting of
$u$. Define $h_j$, $j\in J$, similarly for $B$. The inputs of
$A\vee B$ after flattening binary $\OR$ gates are precisely these two lists, with occurrences retained.
Equation~\eqref{eq:sd-short-gate-certificate} gives
\begin{equation}
 u=\sum_{i\in I}s_i g_i\eqcomma\quad
 v=\sum_{j\in J}t_j h_j\eqcomma \quad
 z=\sum_{i\in I}r_i g_i+\sum_{j\in J}\rho_jh_j\eqcomma
 \label{eq:three-supports}
\end{equation}
with all coefficients of degree at most $\Delta$. (In a singleton case,
the corresponding coefficient is simply~$1$.) We can therefore write
\[
 u(1-z)=\sum_{i\in I}s_i[g_i(1-z)] \eqcomma
 \qquad
 v(1-z)=\sum_{j\in J}t_j[h_j(1-z)]\eqcomma
\]
and moreover 
\[
 z(1-u)(1-v)
 =\sum_{i\in I}r_i(1-v)[g_i(1-u)]
 +\sum_{j\in J}\rho_j(1-u)[h_j(1-v)] \eqperiod
\]
Each bracket is an approximation error as in \eqref{eq:shorthands} associated to $A$, $B$, or $A\lor B$. (In a singleton
case, it is instead $u(1-u)$ or $v(1-v)$, which follow from the Boolean axioms.)
Thus all three polynomials
\[
 u(1-z),\qquad v(1-z),\qquad z(1-u)(1-v)
\]
are derivable in degree at most $4\Delta$ from approximation errors and the Boolean axioms. The elementary identity
\[
 z-u-v+uv
 =z(1-u)(1-v)-u(1-z)-(1-u)v(1-z)
\]
then derives
\[
 \App_{A\vee B}-(\App_A+\App_B-\App_A\App_B) \eqcomma
\]
as desired. The $\AND$ case is entirely analogous: replacing $u, v$, and $z$ by $(1-u), (1-v)$, and $(1-z)$ in the argument above shows that we can derive
\begin{equation}
 \App_{A\wedge B}-\App_A\App_B
 \label{eq:local-and}
\end{equation}
in the same degree.

We now derive \eqref{eq:approximation-errors} by structural induction. For a metavariable $X$, we have
$D_X=0$ by definition. Constants are exact, and negation gives
 $D_{\neg H}=-D_H$. 

Suppose $G=H\vee K$. Put
\[
 u=\App_{\sigma(H)},\quad
 v=\App_{\sigma(K)},\quad
 z=\App_{\sigma(G)},\qquad
 U=H(a),\quad V=K(a).
\]
The argument above shows that we can derive $z-u-v+uv$ from approximation errors and the Boolean axioms in degree $O_{\Rule}(\Delta)$. By induction, we have already
derived $u-U=D_H$ and $v-V=D_K$. Hence
\[
 D_G
 =z-(U+V-UV)
 =z-u-v+uv+(1-v)D_H+(1-U)D_K \eqcomma
\]
which is then also derivable. For $G=H\wedge K$, 
the same conclusion follows from \eqref{eq:local-and} and
\[
 D_G=z-UV=(z-uv)+vD_H+UD_K \eqcomma
\]
completing the induction. In particular, the polynomials
\begin{equation}
    \label{eq:formula-differences}
 D_{F_i}=\App_{L_i}-F_i(a),
 \qquad
 D_F=\App_L-F(a)
\end{equation}
are derivable in degree $O_{\Rule}(\Delta)$. 
 
Let us now derive $E$ from \eqref{eq:formula-differences}. Since each $a_X$ is Boolean-valued, the polynomial $E^*$ in \eqref{eq:sound-rule-polynomial}
vanishes on the Boolean cube by soundness of $\Rule$: no consistent Boolean assignment to the inputs of $F_i$ and $F$
can satisfy every premise but falsify the conclusion. Thus $E^*$ is derivable from the Boolean axioms in degree $O_{\Rule}(\Delta)$. By the telescoping identity
\begin{equation}
 E-E^*
 =(1-\App_L)\sum_{i=1}^r
 \Big(\prod_{h<i}\App_{L_h}\Big)
 D_{F_i}
 \Big(\prod_{h>i}F_h(a)\Big)\\
 -\Big(\prod_{i=1}^rF_i(a)\Big)D_F
 \label{eq:inference-telescope}
\end{equation}
and the degree bounds in \eqref{eq:formula-differences}, we conclude that $E$ is derivable from approximation errors and Boolean axioms in degree $O_{\Rule}(\Delta)$. Note that the same argument includes
zero-premise Frege axioms, in which case the inference error is $1-\App_L$.

\paragraph{Modular recursion axioms.} We now derive the modular recursion axioms using more or less the same reasoning as for Frege rules. 
As in \cref{sec:AC0pFrege}, write
\[
 U=\MOD_{p,a}(\Gamma,B)\eqcomma\quad
 V=\MOD_{p,a}(\Gamma)\eqcomma\quad
 W=\MOD_{p,a-1}(\Gamma)\eqperiod 
\]
In clausal form, the modular recursion axioms then take the form 
\begin{equation}
    \label{eq:clausal-form}
 B\vee\neg U\vee V \eqcomma
 \quad B\vee U\vee\neg V \eqcomma
 \quad \neg B\vee\neg U\vee W \eqcomma
 \quad \neg B\vee U\vee\neg W \eqperiod
\end{equation}
Define
\[
 b=\App_B \eqcomma
 \quad u=\App_U \eqcomma
 \quad v=\App_V \eqcomma
 \quad w=\App_W \eqcomma
 \quad s=\sum_{A\in\Gamma}\App_A\eqperiod
\]
The exact RS construction \eqref{eq:sd-mod-approx} on modular gates gives
\[
 \begin{aligned}
 u&=1-(s+b-a)^{p-1}\eqcomma\\
 v&=1-(s-a)^{p-1}\eqcomma\\
 w&=1-(s-(a-1))^{p-1} \eqperiod
 \end{aligned}
\]
On the Boolean cube, we therefore have
\[
 b=0\implies u=v,\qquad b=1\implies u=w \eqperiod
\]
Hence, the exact evaluation $C(b,u,v,w)$ for each clause $C$ in \eqref{eq:clausal-form} equals $1$ on the Boolean cube.

The same structural induction as for \eqref{eq:formula-differences} derives
\begin{equation}
 \App_C-C(b,u,v,w) 
\end{equation}
in degree $O_{p}(\Delta)$ for each $C$, while $C(b,u,v,w)-1$ follows from the Boolean axioms in degree at most
$3\Delta$. Combining these derives $1-\App_C$, which is the inference error of
this zero-premise axiom. The empty-input modular axioms $\sdMOD_{p,0}(\emptyset)$ and
$\neg\sdMOD_{p,a}(\emptyset)$ for $a\ne0$ are exact and hence have no inference error. 

\paragraph{Input axioms.} 
Consider an input clause $C=\ell_1 \lor \ldots \lor \ell_m$ for literals $\ell_i$ and $m\leq w$. Its polynomial representation is $p_C = \prod_{i=1}^m (1-\ell_i)$, and its approximation errors are $e_{C, i}=\ell_i(1-\App_C)$ (if $m=1$, the approximation is exact and $p_C = 1-\App_C$). We want to derive $1-\App_C$ from $p_C$ and the $e_{C, j}$. We have the identity 
\begin{equation}
    \label{eq:identity-clause1}
1=
\prod_{i=1}^m(1-\ell_i)
+
\sum_{i=1}^m
\Big(\prod_{j<i}(1-\ell_j)\Big)\ell_i \eqperiod
\end{equation}
Multiplying by $1-\App_C$ we see that
\begin{equation}
    \label{eq:identity-clause2}
\begin{aligned}
    1-\App_C
    &= (1-\App_C)\prod_{i=1}^m(1-\ell_i)
+
\sum_{i=1}^m
\Big(\prod_{j<i}(1-\ell_j)\Big)\ell_i(1-\App_C) \\
& = (1-\App_C)p_C + \sum_{i=1}^m
\Big(\prod_{j<i}(1-\ell_j)\Big)e_{C, i} 
\eqcomma
\end{aligned}
\end{equation}
so we can derive $(1-\App_C)$ in degree at most $w + \Delta$ as desired.

By picking $B$ large enough to cover both the Frege rule and modular recursion cases, the lemma follows. \qedhere

\end{proof}

\subsection{Putting Things Together}
\label{subsec:putting-things-together}

With the technical lemmas established, we can now carry out the steps in the proof overview with explicit parameters.

\begin{proof}[Proof of Theorem~\ref{thm:sd-tradeoff}]
Suppose $\refstd=(L_1,\ldots,L_\tau=\bot)$ violates
\eqref{eq:sd-hilbert-rank-sum}. Fix the degrees in~\eqref{eq:sd-degrees} and sample the
RS approximators consistently across all syntactic subformulas of all proof lines. For each conjunction or disjunction $A$ among the subformulas, consider
the error map
\[
 \mult_A:\polyspace{K}/\PC_{\dthresh}(K)\to(\polyspace{\dthresh}/\sdCD)^{m_A} \;;
 \qquad [f]\mapsto([fe_{A,j}])_j \eqcomma
\]
where $e_{A, j}$ is as in \eqref{eq:shorthands}. 
For $i\in [\ell]$, condition on the choices of approximators below depth $i$. The inputs $g_{A,j}$ at a
depth-$i$ gate are now fixed and have degree at most
$d=\Delta_{i-1}$. Since $pd\le2\Delta_i$ and the $\Delta_i$ are
nondecreasing, \eqref{eq:sd-budget} implies
\[
 K+\Delta_i+p\Delta_{i-1}
 \le K+3\Delta_\ell\le \dthresh.
\]
Lemma~\ref{lem:sd-kernel-growth} therefore gives
$\sdE[\sdRank \mult_A\mid\text{lower depths}]\le h(K)\varepsilon_i$.
As in the approximation definition, gates of fan-in smaller than $t_i$
are represented exactly. Their error maps trivially have rank zero, so the same
bound holds. Taking expectations and summing over $i\in [\ell]$,
\[
 \sdE\sum_A\sdRank \mult_A
 \le h(K)\sum_iM_i\varepsilon_i<h(K) \eqperiod
\]
Here $h(K)>0$ because $1\notin\sdCD$. By subadditivity of rank, we have for some complete
choice of approximators that the joint error map $f\mapsto (\mult_A(f))_{\text{subformulas }A}$ has rank below $h(K)$. Thus, by rank-nullity there exists a nonzero class $[f]$ in the joint kernel $\bigcap_{{\text{subformulas }A}}\ker{\mult_A}$ with a representative $f \in \polyspace{K}$.
Explicitly, this means that there exists a degree-at-most-$K$ polynomial $f$ with
\begin{equation}\label{eq:sd-common-errors}
 f\notin\sdCD,\qquad fe_{A,j}\in\sdCD\quad\text{for every }A,j \eqperiod
\end{equation}

Put $p_L=1-\App_L$. For each line $L$ of $\refstd$, we now derive $fp_L\in\sdCD$. 
Input clauses follow from the last assertion of
\cref{lem:sd-inference-errors}, using $K+w+\Delta_\ell\le \dthresh$.
For an inference 
\[
\frac{L_1,\ldots,L_r}{L}
\]
we have, by similar manipulations as in \eqref{eq:algebraic-inference}, the telescoping identity
\begin{equation}\label{eq:sd-multiplied-inference}
 fp_L=fE+
 \sum_{j=1}^r
 \Big(p_L\prod_{v<j}\App_{L_v}\Big)fp_{L_j}
\end{equation}
 where $E=(\prod_j \App_{L_j})p_L$ is the inference error. \cref{lem:sd-inference-errors} and
\eqref{eq:sd-common-errors} give a derivation of $fE$
in degree at most $K+B\Delta_\ell$, and each remaining term is also derivable in this degree for large enough $B$. Hence, we can then derive the conclusion from the already derived premises $fp_{L_j}$.
By Lemma~\ref{lem:sd-inference-errors}, the polynomial on 
the right in \eqref{eq:sd-multiplied-inference} has degree at most
$K+B\Delta_\ell$, where $B$ is the constant in the lemma statement. Zero-premise axioms have $p_L=E$ and are covered by the same argument.

Since the last line in $\refstd$ has $p_{\bot} = 1 -\App_\bot = 1$ we get that
$f\in\sdCD$, contradicting \eqref{eq:sd-common-errors}.
This proves \eqref{eq:sd-hilbert-rank-sum}, and since $\sum_iM_i\le S$,
it also gives $S\max_i\varepsilon_i\ge1$, which is
\eqref{eq:sd-size-bound}.
\end{proof}

\subsection{Setting Parameters}
\label{subsec:sd-parameters}

Let us now derive Theorems~\ref{thm:sd-fixed-depth} and \ref{thm:sd-depth-two} from \cref{thm:sd-tradeoff}. Fix a putative size~$S$, put $\lambda=\ln(2S)$, and choose each $t_i$ and $\Delta_i$ by
\begin{equation}
 t_i=\left\lceil\frac{\lambda}{\eta_K(\Delta_{i-1})}\right\rceil
 \qquad \Delta_i=(p-1)t_i\Delta_{i-1}\eqperiod
 \label{eq:sd-factor-choice}
\end{equation}
With these parameters, 
\[
\rho_K(\Delta_{i-1})^{t_i} = (1-\eta_K (\Delta_{i-1}))^{t_i}\le e^{-\lambda}=1/(2S)\eqcomma 
\]
so $\sum_iM_i\varepsilon_i\le1/2$.
Theorem~\ref{thm:sd-tradeoff} excludes such a refutation whenever
$\Delta_{\ell-1}\le K\le n$ and
$K+\max\{w+\Delta_\ell,B\Delta_\ell\}\le \dthresh$.

For the asymptotic consequences, put $s=n/2-K$. We use
\begin{equation}\label{eq:sd-binomial-ratio}
 \frac{N_n(K-d)}{N_n(K)}
 \ge\exp\!\left(-C\frac{d(s+d+\sqrt n)}n\right),
 \qquad \frac n4\le K-d\le K\le\frac n2,
\end{equation}
for some absolute constant $C$. To verify \eqref{eq:sd-binomial-ratio}, for $n/4\le j\le n/2$ set
$u=n/2-j$ and $m=n/(16(u+\sqrt n))$. Comparing consecutive binomial
coefficients gives $\binom n{j-v}\ge c\binom nj$ for $0\le v\le m$.
Thus $\binom nj/N_n(j)\le C(u+\sqrt n)/n$; this ratio is also bounded
away from one. Sum
$-\ln(N_n(j-1)/N_n(j))$ over $K-d<j\le K$ to obtain
\eqref{eq:sd-binomial-ratio}. 
Substituting in
\eqref{eq:sd-factor-choice} yields the recurrence
\begin{equation}\label{eq:sd-degree-recurrence}
 \Delta_i\le C_p\lambda\Delta_{i-1}
 \exp\!\left(C\frac{\Delta_{i-1}
 (s+\Delta_{i-1}+\sqrt n)}n\right).
\end{equation}

\begin{proof}[Proof of Theorem~\ref{thm:sd-fixed-depth}]
Fix any constant $A>0$ and suppose $S\le n^A$. Choose
\[
 \dthresh=\frac n2-g,\qquad K=\dthresh-w-\sqrt n,
 \qquad\lambda=\ln(2S)=O_A(\log n).
\]
Then $s=g+w+\sqrt n=o(n/(\log n)^{\ell-2})$.
Inductively, \eqref{eq:sd-degree-recurrence} gives
$\Delta_i=O_{p,\ell,A}((\log n)^i)$ for $i\le\ell-1$:
at these steps $s\Delta_{i-1}/n=o(1)$, and the other exponent terms
also tend to zero. At the last step,
$s\Delta_{\ell-1}/n=o(\log n)$, so
\[
 \Delta_\ell=O_{p,\ell,A}((\log n)^\ell)e^{o(\log n)}=n^{o(1)}.
\]
All uses of \eqref{eq:sd-binomial-ratio} have
$K-\Delta_{i-1}\ge n/4$ for large $n$.
Moreover, $B\Delta_\ell=o(\sqrt n)$, so the fact that
$\dthresh-K=w+\sqrt n$ verifies \eqref{eq:sd-budget}. This contradicts
Theorem~\ref{thm:sd-tradeoff}. Since $A$ was arbitrary, the result follows.
\end{proof}

\begin{proof}[Proof of Theorem~\ref{thm:sd-depth-two}]
Choose
\[
 \dthresh=\frac n2-a\frac n{\sqrt k},\qquad
 K=\frac n2-(a+2)\frac n{\sqrt k}.
\]
If $S<\exp(c\sqrt k\ln n)$, then
$\lambda\le2c\sqrt k\ln n$ for large $n$.
The first step has $\eta_K(1)=\Omega_p(1)$, so
$\Delta_1=O_p(\lambda)$. The next step gives
\[
 \Delta_2\le C_p\lambda^2
 \exp\!\left(C_{p,a}\frac\lambda{\sqrt k}
      +C_p\frac{\lambda^2}n+C_p\frac\lambda{\sqrt n}\right).
\]
Choose $c=c(p,a)>0$ so that the first exponent term is at most
$\tfrac14\ln n$; the other two tend to zero. Consequently,
\[
 \Delta_2=O_{p,a}(k(\log n)^2n^{1/4})=o(n/\sqrt k) \eqperiod
\]
Also $w\le k=o(n/\sqrt k)$ and $\Delta_1\le K$.
Since $\dthresh-K=2n/\sqrt k$, the degree budget holds, again contradicting
Theorem~\ref{thm:sd-tradeoff}.
\end{proof}

\bibliographystyle{alpha}

\addcontentsline{toc}{section}{References}
{\small
\bibliography{../refArticles, ../refBooks, ../refLocal}
}

\end{document}